\documentclass{amsart}

\newtheorem{theorem}{Theorem}[section]
\newtheorem{corollary}[theorem]{Corollary}
\newtheorem{lemma}[theorem]{Lemma}

\theoremstyle{definition}
\newtheorem{definition}[theorem]{Definition}

\numberwithin{equation}{section}

\DeclareMathOperator{\len}{len}

\begin{document}

\title{Binary Proportional Pairing Functions}

\author{Matthew P. Szudzik}

\date{November 10, 2018}

\begin{abstract}
A pairing function for the non-negative integers is said to be \emph{binary perfect} if the binary representation of the output is of length $2k$ or less whenever each input has length $k$ or less.  Pairing functions with square shells, such as the Rosenberg-Strong pairing function, are binary perfect.  Many well-known discrete space-filling curves, including the discrete Hilbert curve, are also binary perfect.  The concept of a \emph{binary proportional} pairing function generalizes the concept of a binary perfect pairing function.  Binary proportional pairing functions may be useful in applications where a pairing function is used, and where the function's inputs have lengths differing by a fixed proportion.  In this article, a general technique for constructing a pairing function from any non-decreasing unbounded function is described.  This technique is used to construct a binary proportional pairing function and its inverse.  
\end{abstract}

\maketitle

\section{Introduction}

A bijection from the set $\mathbb{N}^2$ of ordered pairs of non-negative integers to the set $\mathbb{N}$ of non-negative integers is said to be a \emph{pairing function}.  Pairing functions play an important role in computability theory~\cite{Rogers1967}, and have practical applications in computer science~\cite{Rosenberg2003}.  More generally, we say that a bijection from $\mathbb{N}^d$ to $\mathbb{N}$, where $d$ is a positive integer, is an \emph{$d$-tupling function}.  The concept of a $d$-tupling function generalizes the concept of a pairing function to higher dimensions.

For each integer $n>1$ and each non-negative integer $x$, let $\len_n(x)$ denote the number of digits in the base-$n$ representation of $x$.  We say that $\len_n(x)$ is the base-$n$ \emph{length} of $x$.  Formally,
\begin{equation*}
\len_n(x)=\bigl\lceil\,\log_n(x+1)\bigr\rceil,
\end{equation*}
where $\lceil t\rceil$ denotes the ceiling of $t$ for each real number $t$.  By this definition,
\begin{equation}\label{len-identity}
\len_n(x)\leq k\;\quad\text{if and only if}\;\quad x<n^k,
\end{equation}
for all integers $n>1$, and all non-negative integers $x$ and $k$.
%
%

Now, when a $d$-tupling function $f$ is implemented on a computer, integer overflow can be avoided by paying careful attention to the manner in which
%
%
$\len_2\bigl(f(x_1,x_2,\ldots,\linebreak[0]x_d)\bigr)$ depends on $\len_2(x_1)$, $\len_2(x_2)$, \ldots, $\len_2(x_d)$.  With this in mind, we make the following definition.

\begin{definition}\label{perfect}
Let $n$ be any integer greater than $1$.  A $d$-tupling function $f$ is said to be \emph{base-$n$ perfect} if and only if, for all non-negative integers $x_1$, $x_2$, \ldots, $x_d$ and $k$,
\begin{equation}\label{len-x-leq-k}
\len_n(x_1)\leq k\;\And\;\len_n(x_2)\leq k\;\And\;\cdots\;\And\;\len_n(x_d)\leq k
\end{equation}
implies
\begin{equation*}
\len_n\bigl(f(x_1,x_2,\ldots,x_d)\bigr)\leq dk.
\end{equation*}
In the special case where $n=2$, any such function is said to be \emph{binary perfect}.
\end{definition}

An alternate way to characterize the base-$n$ perfect $d$-tupling functions is provided by the following corollary.

\begin{corollary}\label{alt-perfect}
Let $n$ be any integer greater than $1$.  A $d$-tupling function $f$ is base-$n$ perfect if and only if, for all non-negative integers $x_1$, $x_2$, \ldots, $x_d$,
\begin{equation}\label{alt-perfect-inequality}
\len_n\bigl(f(x_1,x_2,\ldots,x_d)\bigr)\leq d\max\bigl(\len_n(x_1),\len_n(x_2),\ldots,\len_n(x_d)\bigr).
\end{equation}
\end{corollary}
\begin{proof}
Let $f$ be any $d$-tupling function and suppose that $f$ is base-$n$ perfect.  Given any non-negative integers $x_1$, $x_2$, \ldots, $x_d$, let
%
%
$k=\max\bigl(\len_n(x_1),\len_n(x_2),\ldots,\linebreak[0]\len_n(x_d)\bigr)$.  It immediately follows from Definition~\ref{perfect} that
\begin{equation*}
\len_n\bigl(f(x_1,x_2,\ldots,x_d)\bigr)\leq d\max\bigl(\len_n(x_1),\len_n(x_2),\ldots,\len_n(x_d)\bigr).
\end{equation*}
Conversely, suppose that inequality~\eqref{alt-perfect-inequality} holds for all non-negative integers $x_1$, $x_2$, \ldots, $x_d$.  Notice that for all non-negative integers $x_1$, $x_2$, \ldots, $x_d$ and $k$, if condition~\eqref{len-x-leq-k} holds, then
\begin{equation*}
\max\bigl(\len_n(x_1),\len_n(x_2),\ldots,\len_n(x_d)\bigr)\leq k.
\end{equation*}
Hence, for all non-negative integers $x_1$, $x_2$, \ldots, $x_d$ and $k$, if condition~\eqref{len-x-leq-k} holds, then it follows from inequality~\eqref{alt-perfect-inequality} that
\begin{equation*}
\len_n\bigl(f(x_1,x_2,\ldots,x_d)\bigr)\leq d\max\bigl(\len_n(x_1),\len_n(x_2),\ldots,\len_n(x_d)\bigr)\leq dk.
\end{equation*}
By definition, $f$ is base-$n$ perfect.
\end{proof} 

Commonly-studied examples of base-$n$ perfect pairing functions, for various bases $n$, are discussed in Sections~\ref{square-shells} and~\ref{space-filling}.  But the concept of a base-$n$ perfect pairing function can also be generalized, as follows.

\begin{definition}\label{proportional}
Let $a$ and $b$ be any positive integers, and let $n$ be any integer greater than $1$.  A pairing function $f$ is said to be \emph{base-$n$ proportional with constants of proportionality $a$ and $b$} if and only if, for all non-negative integers $x$, $y$, and $k$,
\begin{equation*}
\Bigl(\len_n(x)\leq ak\;\And\;\len_n(y)\leq bk\Bigr)\quad\text{implies}\quad\len_n\bigl(f(x,y)\bigr)\leq ak+bk.
\end{equation*}
In the special case where $n=2$, any such function is said to be \emph{binary proportional}.
\end{definition}

In Section~\ref{technique} we introduce a general technique for constructing a pairing function from any non-decreasing unbounded function $g\colon\mathbb{N}\to\mathbb{N}$.  This technique is a variation of Rosenberg's technique for constructing a pairing function that favors a specific shape~\cite{Rosenberg1975b}.  Our technique is then used in Section~\ref{derivation} to construct a pairing function $p_{a,b}$ for each pair of positive integers $a$ and $b$.  For every integer $n>1$, this pairing function is base-$n$ proportional with constants of proportionality $a$ and $b$.  In particular, for each $(x,y)\in\mathbb{N}^2$ we define
\begin{equation*}
p_{a,b}(x,y)=\begin{cases}
y\bigl\lfloor\sqrt[b]{y\rule{0pt}{6pt}}\,\bigr\rfloor^a+x &\text{\quad if $\bigl\lfloor\sqrt[b]{y\rule{0pt}{6pt}}\,\bigr\rfloor>\bigl\lfloor\sqrt[a]{x}\,\bigr\rfloor$}\\
x\Bigl(\bigl\lfloor\sqrt[a]{x}\,\bigr\rfloor+1\Bigr)^{\!b}+y &\text{\quad otherwise}\rule{0pt}{20pt}
\end{cases},
\end{equation*}
%
%
where $\lfloor t\rfloor$ denotes the floor of $t$ for each real number $t$.  The function's inverse is given by
\begin{equation*}
p_{a,b}^{-1}(z)=\begin{cases}
\biggl(z\bmod m^a\;,\;\Bigl\lfloor\,\dfrac{z}{m^a}\Bigr\rfloor\biggr) &\text{\quad if $z<m^a(m+1)^b$}\\
\biggl(\Bigl\lfloor\dfrac{z}{(m+1)^b\rule{0pt}{9pt}}\Bigr\rfloor\;,\;z\bmod (m+1)^b\biggr) &\text{\quad otherwise}\rule{0pt}{21pt}
\end{cases},
\end{equation*}
where $m=\bigl\lfloor z^{1/(a+b)}\bigr\rfloor$ for each $z\in\mathbb{N}$.

By composing functions of the form $p_{a,b}$, various sorts of $d$-tupling functions can be constructed.  For example,
\begin{equation*}
p_{2,1}\bigl(p_{1,1}(x_1,x_2),x_3\bigr)
\end{equation*}
is a base-$n$ perfect $3$-tupling function for all integers $n>1$,
\begin{equation*}
p_{3,1}\Bigl(p_{2,1}\bigl(p_{1,1}(x_1,x_2),x_3\bigr),x_4\Bigr)
\end{equation*}
is a base-$n$ perfect $4$-tupling function for all integers $n>1$, and so on.  More practically, if $x_1$ is a $32$-bit unsigned integer,\footnote{
The phrase ``32-bit unsigned integer'', for example, refers to any non-negative integer $x$ such that $\len_2(x)\leq32$.
} if $x_2$ is a $48$-bit unsigned integer, and if $x_3$ is a $64$-bit unsigned integer, then
\begin{equation*}
p_{5,4}\bigl(p_{2,3}(x_1,x_2),x_3\bigr)
\end{equation*}
is a $144$-bit unsigned integer.  This fact follows directly from Definition~\ref{proportional}, since
\begin{multline*}
\Bigl(\len_2(x_1)\leq32=2\cdot16\;\And\;\len_2(x_2)\leq48=3\cdot16\Bigr)\\
\text{implies}\;\quad\len_2\bigl(p_{2,3}(x_1,x_2)\bigr)\leq32+48
\end{multline*}
and
\begin{multline*}
\Bigl(\len_2\bigl(p_{2,3}(x_1,x_2)\bigr)\leq32+48=5\cdot16\;\And\;\len_2(x_3)\leq64=4\cdot16\Bigr)\\
\text{implies}\;\quad\len_2\Bigl(p_{5,4}\bigl(p_{2,3}(x_1,x_2),x_3\bigr)\Bigr)\leq144=(32+48)+64.
\end{multline*}

Now, given any positive integers $a$, $b$, and $c$, if $c>1$ then
\begin{equation*}
p_{a,b}(1,0)=2^b\ne2^{bc}=p_{ac,bc}(1,0).
\end{equation*}
An immediate consequence is that $p_{a,b}$ and $p_{ac,bc}$ are distinct functions if $c>1$.
%
%
Nevertheless, we have the following theorem.

\begin{theorem}
Let $a$ and $b$ be any positive integers, and let $n$ be any integer greater than $1$.  If a base-$n$ proportional pairing function $f$ has constants of proportionality $a$ and $b$, then $f$ also has constants of proportionality $ac$ and $bc$, for all positive integers $c$.
\end{theorem}
\begin{proof}
Suppose that $f$ is a base-$n$ proportional pairing function with constants of proportionality $a$ and $b$.  Consider any positive integer $c$ and any non-negative integer $k^\prime$.  By Definition~\ref{proportional},
\begin{equation*}
\Bigl(\len_n(x)\leq ak\;\And\;\len_n(y)\leq bk\Bigr)\quad\text{implies}\quad\len_n\bigl(f(x,y)\bigr)\leq ak+bk
\end{equation*}
for all non-negative integers $x$, $y$, and $k$.  But letting $k=ck^\prime$, we have that
\begin{equation*}
\Bigl(\len_n(x)\leq(ac)k^\prime\;\And\;\len_n(y)\leq(bc)k^\prime\Bigr)\quad\text{implies}\quad\len_n\bigl(f(x,y)\bigr)\leq(ac)k^\prime+(bc)k^\prime
\end{equation*}
for all non-negative integers $x$, $y$, and $k^\prime$.  It then follows from Definition~\ref{proportional} that $f$ is base-$n$ proportional with constants of proportionality $ac$ and $bc$.
\end{proof}

If we are only interested in finding a base-$n$ proportional pairing function with constants of proportionality $ac$ and $bc$, where $a$, $b$, and $c$ are positive integers, then according to the theorem, it is sufficient to consider the function $p_{a,b}$.  For this reason, we often restrict our attention to functions of the form $p_{a,b}$ where $\gcd(a,b)=1$.

\section{Pairing Functions With Square Shells}\label{square-shells}

Given any $d$-tupling function $f\colon\mathbb{N}^d\to\mathbb{N}$, a function $s\colon\mathbb{N}^d\to\mathbb{N}$ is said to be a \emph{shell numbering} for $f$ if and only if
\begin{equation}\label{shell-def}
s(\mathbf{x})<s(\mathbf{y})\quad\text{implies}\quad f(\mathbf{x})<f(\mathbf{y})
\end{equation}
for all $\mathbf{x}$ and $\mathbf{y}$ in $\mathbb{N}^d$.  Pairing functions that have $\max(x,y)$ as a shell numbering are said to have \emph{square shells}.  The Rosenberg-Strong pairing function, given by the formula
\begin{equation*}
r_2(x,y)=\bigl(\max(x,y)\bigr)^2+\max(x,y)+x-y,
\end{equation*}
is the most well-known example of a pairing function with square shells.\footnote{
The functions originally described by Rosenberg and Strong~\cite{Rosenberg1972,Rosenberg1974} were bijections from $\mathbb{P}^d$ to $\mathbb{P}$, where $\mathbb{P}$ denotes the set of all positive integers.  For this article, we have translated these functions to the non-negative integers, and reversed the order of the arguments.
}  This function is illustrated in Figure~\ref{rosenberg-strong},
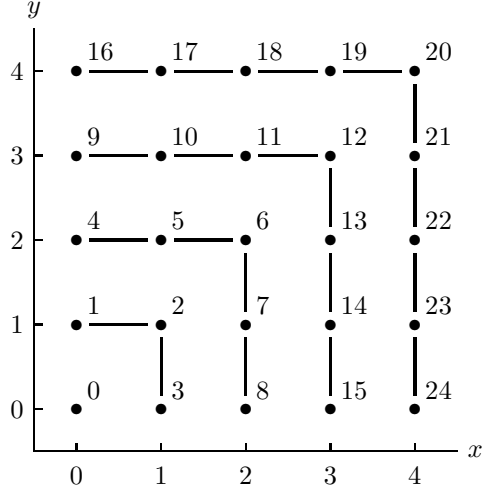
\begin{figure}
\centering
\begin{picture}(179,184)(-25,-29)
%
%
\multiput(0,0)(0,32){5}{\multiput(0,0)(32,0){5}{\circle*{4}}}
%
\thinlines
\put(-16,-16){\line(1,0){160}}
\put(-16,-16){\line(0,1){160}}
%
\put(148,-16){\makebox(0,0)[l]{$x$}}
\put(-16,148){\makebox(0,0)[b]{$y$}}
%
\thinlines
\multiput(0,-16)(32,0){5}{\line(0,1){3}}
\multiput(-16,0)(0,32){5}{\line(1,0){3}}
%
\put(0,-22){\makebox(0,0)[t]{0}}
\put(32,-22){\makebox(0,0)[t]{1}}
\put(64,-22){\makebox(0,0)[t]{2}}
\put(96,-22){\makebox(0,0)[t]{3}}
\put(128,-22){\makebox(0,0)[t]{4}}
%
\put(-20,0){\makebox(0,0)[r]{0}}
\put(-20,32){\makebox(0,0)[r]{1}}
\put(-20,64){\makebox(0,0)[r]{2}}
\put(-20,96){\makebox(0,0)[r]{3}}
\put(-20,128){\makebox(0,0)[r]{4}}
%
\thicklines
\put(5,32){\line(1,0){22}}
\put(32,27){\line(0,-1){22}}
\put(5,64){\line(1,0){22}}
\put(37,64){\line(1,0){22}}
\put(64,59){\line(0,-1){22}}
\put(64,27){\line(0,-1){22}}
\put(5,96){\line(1,0){22}}
\put(37,96){\line(1,0){22}}
\put(69,96){\line(1,0){22}}
\put(96,91){\line(0,-1){22}}
\put(96,59){\line(0,-1){22}}
\put(96,27){\line(0,-1){22}}
\put(5,128){\line(1,0){22}}
\put(37,128){\line(1,0){22}}
\put(69,128){\line(1,0){22}}
\put(101,128){\line(1,0){22}}
\put(128,123){\line(0,-1){22}}
\put(128,91){\line(0,-1){22}}
\put(128,59){\line(0,-1){22}}
\put(128,27){\line(0,-1){22}}
%
\put(4,4){\makebox(0,0)[bl]{0}}
\put(4,36){\makebox(0,0)[bl]{1}}
\put(36,36){\makebox(0,0)[bl]{2}}
\put(36,4){\makebox(0,0)[bl]{3}}
\put(4,68){\makebox(0,0)[bl]{4}}
\put(36,68){\makebox(0,0)[bl]{5}}
\put(68,68){\makebox(0,0)[bl]{6}}
\put(68,36){\makebox(0,0)[bl]{7}}
\put(68,4){\makebox(0,0)[bl]{8}}
\put(4,100){\makebox(0,0)[bl]{9}}
\put(36,100){\makebox(0,0)[bl]{10}}
\put(68,100){\makebox(0,0)[bl]{11}}
\put(100,100){\makebox(0,0)[bl]{12}}
\put(100,68){\makebox(0,0)[bl]{13}}
\put(100,36){\makebox(0,0)[bl]{14}}
\put(100,4){\makebox(0,0)[bl]{15}}
\put(4,132){\makebox(0,0)[bl]{16}}
\put(36,132){\makebox(0,0)[bl]{17}}
\put(68,132){\makebox(0,0)[bl]{18}}
\put(100,132){\makebox(0,0)[bl]{19}}
\put(132,132){\makebox(0,0)[bl]{20}}
\put(132,100){\makebox(0,0)[bl]{21}}
\put(132,68){\makebox(0,0)[bl]{22}}
\put(132,36){\makebox(0,0)[bl]{23}}
\put(132,4){\makebox(0,0)[bl]{24}}
\end{picture}
\caption{The Rosenberg-Strong pairing function $r_2(x,y)$.  To make the sequence of points $r_2^{-1}(0)$, $r_2^{-1}(1)$, $r_2^{-1}(2)$, \ldots\ more visually apparent, line segments have been drawn between some of the points that occur consecutively in the sequence.}
\label{rosenberg-strong}
\end{figure}
and its inverse is given by
\begin{equation*}
r_2^{-1}(z)=\begin{cases}
\bigl(z-m^2,m\bigr) &\text{\quad if $z-m^2<m$}\\
\bigl(m,m^2+2m-z\bigr) &\text{\quad otherwise}\rule{0pt}{15pt}
\end{cases},
\end{equation*}
where $m=\bigl\lfloor\sqrt{z}\,\bigr\rfloor$ for each $z\in\mathbb{N}$.

Rosenberg and Strong also generalized their pairing function to higher dimensions.  In particular, the Rosenberg-Strong $d$-tupling function can be defined recursively so that $r_1(x_1)=x_1$, and so that for each integer $d>1$,
\begin{equation*}
r_d(x_1,\ldots,x_{d-1},x_d)=r_{d-1}(x_1,\ldots,x_{d-1})+m^d+(m-x_d)\bigl((m+1)^{d-1}-m^{d-1}\bigr),
\end{equation*}
where $m=\max(x_1,\ldots,x_{d-1},x_d)$.  The inverse of this function is defined so that $r_1^{-1}(z)=z$, and so that for each integer $d>1$,
\begin{equation}\label{rd-inv}
r_d^{-1}(z)=\Bigl(r_{d-1}^{-1}\bigl(z-m^d-(m-x_d)((m+1)^{d-1}-m^{d-1})\bigr),x_d\Bigr),
\end{equation}
where
\begin{equation*}
x_d=m-\Biggl\lfloor\frac{\max\bigl(0,z-m^d-m^{d-1}\bigr)}{(m+1)^{d-1}-m^{d-1\rule{0pt}{6pt}}}\Biggr\rfloor
\end{equation*}
and  $m=\bigl\lfloor\sqrt[d]{z}\,\bigr\rfloor$.  Note that in equation~\eqref{rd-inv} we use the set-theoretic convention~\cite{Enderton1977} that $(x_1,x_2,x_3,x_4)$ is an abbreviation for $\bigl(\bigl((x_1,x_2),x_3\bigr),x_4\bigr)$, for example.

Generalizing the concept of square shells to higher dimensions, any $d$-tupling function that has $\max(x_1,x_2,\ldots,x_d)$ as a shell numbering is said to have \emph{cubic shells}.  By this definition, every pairing function with square shells also has cubic shells.  Moreover, we have the following theorem~\cite{Szudzik2018}.

\begin{theorem}\label{cubic}
Let $f$ be any $d$-tupling function and let $n$ be any integer greater than $1$.  If $f$ has cubic shells then $f$ is base-$n$ perfect.
\end{theorem}

The Rosenberg-Strong $d$-tupling function is an example of a $d$-tupling function with cubic shells.  Therefore, by the theorem, $r_d$ is base-$n$ perfect for all integers $n>1$.  But the converse of Theorem~\ref{cubic} does not hold---there are examples of pairing functions that are base-$n$ perfect for some integer $n>1$, but that do not have cubic shells.  In fact, the discrete space-filling curves of Hilbert and Peano, discussed in the next section, are two such examples.

Given any integer $n>1$, we say that a $d$-tupling function $f$ has \emph{base-$n$ shells} if and only if $f$ has
\begin{equation*}
\max\bigl(\len_n(x_1),\len_n(x_2),\ldots,\len_n(x_d)\bigr)
\end{equation*}
as a shell numbering.  As will be seen in Theorem~\ref{base-n-shells}, the $d$-tupling functions with base-$n$ shells are exactly those $d$-tupling functions that are base-$n$ perfect.  But first, define 
\begin{equation}\label{u-def}
U_s^{<k}=\bigl\{\,\mathbf{y}\in\mathbb{N}^d\;:\;s(\mathbf{y})<k\,\bigr\}
\end{equation}
for each function $s\colon\mathbb{N}^d\to\mathbb{N}$ and each non-negative integer $k$.  The following theorem~\cite{Szudzik2018} relates $U_s^{<k}$ to shell numberings.

\begin{theorem}\label{upper-and-lower-shell-numbering}
Let $f\colon\mathbb{N}^d\to\mathbb{N}$ be any $d$-tupling function.  A function $s\colon\mathbb{N}^d\to\mathbb{N}$ is a shell numbering for $f$ if and only if, for all $\mathbf{x}\in\mathbb{N}^d$,
\begin{equation*}
\bigl\lvert U_s^{<s(\mathbf{x})}\bigr\rvert\leq f(\mathbf{x})<\bigl\lvert U_s^{<s(\mathbf{x})+1}\bigr\rvert.
\end{equation*}
\end{theorem}

But in the special case where $U_s^{<k}$ is a finite set for each non-negative integer $k$, the following lemma allows us to simplify the theorem.

\begin{lemma}\label{upper-implies-lower}
Let $f\colon\mathbb{N}^d\to\mathbb{N}$ be any $d$-tupling function, and let $s\colon\mathbb{N}^d\to\mathbb{N}$ be any function such that $U_s^{<k}$ is a finite set for each $k\in\mathbb{N}$.  If
\begin{equation}\label{upper-bound}
f(\mathbf{x})<\bigl\lvert U_s^{<s(\mathbf{x})+1}\bigr\rvert
\end{equation}
for all $\mathbf{x}\in\mathbb{N}^d$, then
\begin{equation*}
\bigl\lvert U_s^{<s(\mathbf{x})}\bigr\rvert\leq f(\mathbf{x})
\end{equation*}
for all $\mathbf{x}\in\mathbb{N}^d$.
\end{lemma}
\begin{proof}
Suppose that inequality~\eqref{upper-bound} holds for all $\mathbf{x}\in\mathbb{N}^d$.  Now consider any $k\in\mathbb{N}$ and any $\mathbf{x}\in U_s^{<k+1}$.  By the definition of $U_s^{<k+1}$,
 $s(\mathbf{x})<k+1$.  Therefore, $s(\mathbf{x})\leq k$.  And by inequality~\eqref{upper-bound},
\begin{equation}\label{upper-bound-k}
f(\mathbf{x})<\bigl\lvert U_s^{<s(\mathbf{x})+1}\bigr\rvert\leq\bigl\lvert U_s^{<k+1}\bigr\rvert.
\end{equation}
Now define
\begin{equation}\label{j-def}
J_s^{<k+1}=\bigl\{\,f(\mathbf{y})\;:\;\mathbf{y}\in U_s^{<k+1}\,\bigr\}.
\end{equation}
Because $f\colon\mathbb{N}^d\to\mathbb{N}$ is a bijection, $\bigl\lvert J_s^{<k+1}\bigr\rvert=\bigl\lvert U_s^{<k+1}\bigr\rvert$.  Therefore, by inequality~\eqref{upper-bound-k}, $f(\mathbf{x})<\bigl\lvert J_s^{<k+1}\bigr\rvert$.  But $f(\mathbf{x})\in J_s^{<k+1}$ because $\mathbf{x}\in U_s^{<k+1}$.  We have just shown that every member of $J_s^{<k+1}$ is less than $\bigl\lvert J_s^{<k+1}\bigr\rvert$.  And since $\bigl\lvert J_s^{<k+1}\bigr\rvert=\bigl\lvert U_s^{<k+1}\bigr\rvert$ is finite, it immediately follows that
\begin{equation}\label{j-equation}
J_s^{<k+1}=\bigl\{0,1,2,\ldots,\bigl\lvert J_s^{<k+1}\bigr\rvert-1\bigr\}
\end{equation}
for all $k\in\mathbb{N}$.

Now consider any $\mathbf{z}\in\mathbb{N}^d$.  If $s(\mathbf{z})=0$ then
\begin{equation*}
f(\mathbf{z})\geq0=\bigl\lvert U_s^{<s(\mathbf{z})}\bigr\rvert
\end{equation*}
because $U_s^{<s(\mathbf{z})}$ is empty.  Alternatively, if $s(\mathbf{z})>0$ then let $k=s(\mathbf{z})-1$.  By equation~\eqref{j-equation}, $J_s^{<s(\mathbf{z})}=\bigl\{0,1,2,\ldots,\bigl\lvert J_s^{<s(\mathbf{z})}\bigr\rvert-1\bigr\}$.  But by the definition of $J_s^{<s(\mathbf{z})}$, $J_s^{<s(\mathbf{z})}=\bigl\{f(\mathbf{y}):\mathbf{y}\in U_s^{<s(\mathbf{z})}\bigr\}$, and $\mathbf{z}\notin U_s^{<s(\mathbf{z})}$ by the definition of $U_s^{<s(\mathbf{z})}$.  Therefore, $f(\mathbf{z})\notin J_s^{<s(\mathbf{z})}$, and it must be the case that
\begin{equation*}
f(\mathbf{z})\geq\bigl\lvert J_s^{<s(\mathbf{z})}\bigr\rvert=\bigl\lvert U_s^{<s(\mathbf{z})}\bigr\rvert.
\end{equation*}
In either case, we have shown that $f(\mathbf{z})\geq\bigl\lvert U_s^{<s(\mathbf{z})}\bigr\rvert$, and this holds for all $\mathbf{z}\in\mathbb{N}^d$.
\end{proof}

\begin{corollary}\label{upper-shell-numbering}
Let $f\colon\mathbb{N}^d\to\mathbb{N}$ be any $d$-tupling function, and let $s\colon\mathbb{N}^d\to\mathbb{N}$ be any function such that $U_s^{<k}$ is a finite set for each $k\in\mathbb{N}$.  Then, $s$ is a shell numbering for $f$ if and only if, for all $\mathbf{x}\in\mathbb{N}^d$,
\begin{equation*}
f(\mathbf{x})<\bigl\lvert U_s^{<s(\mathbf{x})+1}\bigr\rvert.
\end{equation*}
\end{corollary}
\begin{proof}
This is an immediate consequence of Theorem~\ref{upper-and-lower-shell-numbering} and Lemma~\ref{upper-implies-lower}.
\end{proof}

The corollary can now be applied to prove the following theorem.

\begin{theorem}\label{base-n-shells}
Let $n$ be any integer greater than $1$.  A $d$-tupling function is base-$n$ perfect if and only if it has base-$n$ shells.
\end{theorem}
\begin{proof}
Let $s(x_1,\ldots,x_d)=\max\bigl(\len_n(x_1),\ldots,\len_n(x_d)\bigr)$ for each $(x_1,\ldots,x_d)\in\mathbb{N}^d$.  Now consider any $\mathbf{x}\in\mathbb{N}^d$.  By definition,
\begin{equation*}
U_s^{<s(\mathbf{x})+1}=\bigl\{\,(y_1,\ldots,y_d)\in\mathbb{N}^d\;:\;\max\bigl(\len_n(y_1),\ldots,\len_n(y_d)\bigr)<s(\mathbf{x})+1\,\bigr\}.
\end{equation*}
But $\max\bigl(\len_n(y_1),\ldots,\len_n(y_d)\bigr)<s(\mathbf{x})+1$ if and only if
\begin{equation*}
\len_n(y_i)\leq s(\mathbf{x})\quad\text{for all}\quad i\in\{1,\ldots,d\}.
\end{equation*}
Hence, by condition~\eqref{len-identity},
\begin{equation*}
U_s^{<s(\mathbf{x})+1}=\bigl\{\,(y_1,\ldots,y_d)\in\mathbb{N}^d\;:\;\text{$y_i<n^{s(\mathbf{x})}$ for all $i\in\{1,\ldots,d\}$}\,\bigr\}.
\end{equation*}
It immediately follows that $\bigl\lvert U_s^{<s(\mathbf{x})+1}\bigr\rvert=n^{s(\mathbf{x})d}$ for all $\mathbf{x}\in\mathbb{N}^d$.

Now consider any $d$-tupling function $f$.  By definition, $f$ has base-$n$ shells if and only if $s$ is a shell numbering for $f$.  Hence, by Corollary~\ref{upper-shell-numbering}, $f$ has base-$n$ shells if and only if $f(\mathbf{x})<n^{s(\mathbf{x})d}$ for all $\mathbf{x}\in\mathbb{N}^d$.  But by condition~\eqref{len-identity}, this is true if and only if $\len_n\bigl(f(\mathbf{x})\bigr)\leq s(\mathbf{x})d$ for all $\mathbf{x}\in\mathbb{N}^d$.  And by Corollary~\ref{alt-perfect}, this is true if and only if $f$ is base-$n$ perfect.
\end{proof}

The $d$-tupling functions with cubic shells may also be characterized in the following manner.

\begin{theorem}
Let $f$ be any $d$-tupling function.  Then $f$ has cubic shells if and only if, for all integers $n>1$, $f$ has base-$n$ shells.
\end{theorem}
\begin{proof}
Suppose that $f$ has cubic shells.  By Theorem~\ref{cubic}, $f$ is base-$n$ perfect for all integers $n>1$.  Then by Theorem~\ref{base-n-shells}, $f$ has base-$n$ shells for all integers $n>1$.

Alternatively, suppose that $f$ does not have cubic shells.  By condition~\eqref{shell-def}, there exist points $(x_1,\ldots,x_d)\in\mathbb{N}^d$ and $(y_1,\ldots,y_d)\in\mathbb{N}^d$ such that
\begin{equation*}
\max(x_1,\ldots,x_d)<\max(y_1,\ldots,y_d)\quad\And\quad f(x_1,\ldots,x_d)\geq f(y_1,\ldots,y_d).
\end{equation*}
Let $m=\max(x_1,\ldots,x_d)$.  There are two cases to consider.
\begin{description}
\item[Case 1] If $m=0$ then $x_i=0$ for each $i\in\{1,\ldots,d\}$.  Therefore, $\len_2(x_i)=0$ for each $i\in\{1,\ldots,d\}$.  But $\max(y_1,\ldots,y_d)>0$, so there must exist an index $j\in\{1,\ldots,d\}$ such that $\len_2(y_j)>0$.  Hence,
\begin{equation*}
\max\bigl(\len_2(x_1),\ldots,\len_2(x_d)\bigr)=0<\len_2(y_j)\leq\max\bigl(\len_2(y_1),\ldots,\len_2(y_d)\bigr).
\end{equation*}
\item[Case 2] If $m>0$ then consider the base-$(m+1)$ lengths of $x_1$, \ldots, $x_d$ and $y_1$, \ldots, $y_d$.  Notice that $x_i<m+1$ for each $i\in\{1,\ldots,d\}$.  By condition~\eqref{len-identity}, this implies that $\len_{m+1}(x_i)\leq1$ for each $i\in\{1,\ldots,d\}$.  And since $\max(y_1,\ldots,y_d)>m$, there must exist an index $j\in\{1,\ldots,d\}$ such that $y_j\geq m+1$.  Then by condition~\eqref{len-identity}, $\len_{m+1}(y_j)>1$.  Hence,
\begin{align*}
\max\bigl(\len_{m+1}(x_1),\ldots,\len_{m+1}(x_d)\bigr)\leq1&<\len_{m+1}(y_j)\\
&\leq\max\bigl(\len_{m+1}(y_1),\ldots,\len_{m+1}(y_d)\bigr).
\end{align*}
\end{description}
In either case, there exists an integer $n>1$ such that
\begin{equation*}
\max\bigl(\len_n(x_1),\ldots,\len_n(x_d)\bigr)<\max\bigl(\len_n(y_1),\ldots,\len_n(y_d)\bigr)
\end{equation*}
and
\begin{equation*}
f(x_1,\ldots,x_d)\geq f(y_1,\ldots,y_d).
\end{equation*}
That is, there exists an integer $n>1$ such that $f$ does not have base-$n$ shells.
\end{proof}

\section{Discrete Space-Filling Curves}\label{space-filling}

We use the following conventions in this section.  For any non-negative integer $n$, a bijection from $\{0,1,\ldots,n-1\}$ to $\{0,1,\ldots,n-1\}$ is said to be a \emph{permutation on $n$ symbols}.  The set of all permutations on $n$ symbols is $S_n$.  We use the two-row notation
\begin{equation*}
\left(\begin{matrix}
0 & 1 & \cdots & n-1\\
y_0 & y_1 & \cdots & y_{n-1}
\end{matrix}\right),
\end{equation*} 
where $y_0$, $y_1$, \ldots, $y_{n-1}$ are non-negative integers, to denote the permutation $\sigma\in S_n$ that satisfies the equations
\begin{equation*}
\sigma(0)=y_0,\qquad\sigma(1)=y_1,\qquad\ldots,\qquad\sigma(n-1)=y_{n-1}.
\end{equation*}
For any sets $A$, $B$, and $C$, and any functions $g\colon A\to B$ and $f\colon B\to C$, we use $f\circ g$ the denote the function from $A$ to $C$ that satisfies the equation $f\circ g(x)=f\bigl(g(x)\bigr)$ for all $x\in A$.  We say that $f\circ g$ is the \emph{composition} of $f$ with $g$.  The identity function is denoted $I$.  We use $\sigma^k$, where $\sigma\in S_n$ and $k\in\mathbb{N}$, to denote the function
\begin{equation*}
\underbrace{\sigma\circ\sigma\circ\cdots\circ\sigma}_\text{$k$ many $\sigma$'s}\circ I.
\end{equation*}
If $k=1$, $\sigma^{-k}$ is the inverse of $\sigma$.  Otherwise, $\sigma^{-k}$ is the function $(\sigma^{-1})^k$.

Now, any continuous surjection from the unit interval $[0,1]$ to the $d$-dimensional unit cube $[0,1]^d$, where $d$ is an integer greater than $1$, is said to be a \emph{space-filling curve}.\footnote{
Sagan~\cite{Sagan1994} provides a more general definition for the notion of a space-filling curve.
}  The first space-filling curves to be discovered were the \emph{continuous Peano curve}~\cite{Peano1890} and the \emph{continuous Hilbert curve}~\cite{Hilbert1891}, both of which are functions from $[0,1]$ to $[0,1]^2$.  If $F\colon[0,1]\to[0,1]^2$ is the continuous Peano curve, then there exists a unique pairing function $f\colon\mathbb{N}^2\to\mathbb{N}$ such that
\begin{equation*}
F\left(\frac{z+1/2}{9^k}\right)=3^{-k}\Bigl(f^{-1}(z)+\bigl(\tfrac{1}{2},\tfrac{1}{2}\bigl)\Bigr)
\end{equation*}
for all non-negative integers $z$ and $k$ such that $z<9^k$.
This function $f$ is said to be the \emph{discrete Peano curve}.  More generally, we make the following definition.

\begin{definition}\label{discrete}
Let $n$ be any integer greater than $1$.  A function $f\colon\mathbb{N}^2\to\mathbb{N}$ is said to be a \emph{base-$n$ discrete space-filling curve} if and only if there exist permutations $\tau$, $\sigma_0$, $\sigma_1$, \ldots, $\sigma_{n^2-1}$ in $S_{n^2}$ such that the following three conditions hold.
\begin{enumerate}
\item[(a)] $\tau(0)=0$ and $\sigma_0(0)=0$.
\item[(b)] $f(0,0)=0$.
\item[(c)] For all non-negative integers $x$, $y$, and $z$ such that $(x,y)\ne(0,0)$, $f(x,y)=z$ if and only if
\begin{equation}\label{discrete-equations}
\begin{aligned}
\sigma_0^{-(m-1)}\circ\tau\bigl(nx_{m-1}+y_{m-1}\bigr)&=z_{m-1},\\
\sigma_{z_{m-1}}\circ\sigma_0^{-(m-1)}\circ\tau\bigl(nx_{m-2}+y_{m-2}\bigr)&=z_{m-2},\\
\sigma_{z_{m-2}}\circ\sigma_{z_{m-1}}\circ\sigma_0^{-(m-1)}\circ\tau\bigl(nx_{m-3}+y_{m-3}\bigr)&=z_{m-3},\\
&\;\;\vdots\\
\sigma_{z_1}\circ\cdots\circ\sigma_{z_{m-2}}\circ\sigma_{z_{m-1}}\circ\sigma_0^{-(m-1)}\circ\tau\bigl(nx_0+y_0\bigr)&=z_0,
\end{aligned}
\end{equation}
where $m=\max\bigl(\len_n(x),\len_n(y)\bigr)$, where
\begin{align*}
x&=x_{m-1}n^{m-1}+x_{m-2}n^{m-2}+\cdots+x_0n^0,\\
y&=y_{m-1}n^{m-1}+y_{m-2}n^{m-2}+\cdots+y_0n^0
\end{align*}
are base-$n$ expansions of $x$ and $y$, and where
\begin{equation*}
z=z_{m-1}\bigl(n^2\bigr)^{m-1}+z_{m-2}\bigl(n^2\bigr)^{m-2}+\cdots+z_0\bigl(n^2\bigr)^0
\end{equation*}
is a base-$n^2$ expansion of $z$.
\end{enumerate}
\end{definition}

If $f\colon\mathbb{N}^2\to\mathbb{N}$ is a base-$n$ discrete space-filling curve, then the corresponding permutations $\tau$, $\sigma_0$, $\sigma_1$, \ldots, $\sigma_{n^2-1}$ uniquely determine the value of $f(x,y)$ for each $(x,y)\in\mathbb{N}^2$.  In particular, equations~\eqref{discrete-equations} can be used to calculate a base-$n^2$ expansion of $f(x,y)$, given the appropriate base-$n$ expansions of $x$ and $y$.  Equations~\eqref{discrete-equations} can also be inverted, as follows:
\begin{equation}\label{discrete-equations-inv}
\begin{aligned}
(x_{m-1},y_{m-1})&=\delta_n^{-1}\circ\tau^{-1}\circ\sigma_0^{m-1}\bigl(z_{m-1}\bigr),\\
(x_{m-2},y_{m-2})&=\delta_n^{-1}\circ\tau^{-1}\circ\sigma_0^{m-1}\circ\sigma_{z_{m-1}}^{-1}\bigl(z_{m-2}\bigr),\\
(x_{m-3},y_{m-3})&=\delta_n^{-1}\circ\tau^{-1}\circ\sigma_0^{m-1}\circ\sigma_{z_{m-1}}^{-1}\circ\sigma_{z_{m-2}}^{-1}\bigl(z_{m-3}\bigr),\\
&\;\;\vdots\\
(x_0,y_0)&=\delta_n^{-1}\circ\tau^{-1}\circ\sigma_0^{m-1}\circ\sigma_{z_{m-1}}^{-1}\circ\sigma_{z_{m-2}}^{-1}\circ\cdots\circ\sigma_{z_1}^{-1}\bigl(z_0\bigr),
\end{aligned}
\end{equation}
where $\delta_n$ is the bijection from $\{0,1,\ldots,n-1\}^2$ to $\{0,1,\ldots,n^2-1\}$ that is given by the formula $\delta_n(x,y)=nx+y$.

\begin{theorem}\label{discrete-perfect}
If $f\colon\mathbb{N}^2\to\mathbb{N}$ is a base-$n$ discrete space-filling curve, where $n$ is an integer greater than $1$, then $f$ is a base-$n$ perfect pairing function.
\end{theorem}
\begin{proof}
Suppose that $f\colon\mathbb{N}^2\to\mathbb{N}$ is a base-$n$ discrete space-filling curve, and consider any $(x,y)\in\mathbb{N}^2$.  Let $m=\max\bigl(\len_n(x),\len_n(y)\bigr)$ and suppose that $(x,y)=(0,0)$.  Then, $m=0$.  And by Definition~\ref{discrete}(b), $f(x,y)=0$.  Therefore,
\begin{equation*}
\len_{n^2}\bigl(f(x,y)\bigr)=\len_{n^2}\bigl(0\bigr)=0=m.
\end{equation*}
Alternatively, suppose that $(x,y)\ne(0,0)$ and let
\begin{align*}
x&=x_{m-1}n^{m-1}+x_{m-2}n^{m-2}+\cdots+x_0n^0,\\
y&=y_{m-1}n^{m-1}+y_{m-2}n^{m-2}+\cdots+y_0n^0
\end{align*}
be base-$n$ expansions of $x$ and $y$.  Note that $m>0$ because $(x,y)\ne(0,0)$.  By Definition~\ref{discrete},
\begin{equation*}
f(x,y)=z_{m-1}\bigl(n^2\bigr)^{m-1}+z_{m-2}\bigl(n^2\bigr)^{m-2}+\cdots+z_0\bigl(n^2\bigr)^0,
\end{equation*}
where $z_{m-1}$, $z_{m-2}$, \ldots, $z_0$ are given by equations~\eqref{discrete-equations}.  In particular,
\begin{equation*}
\sigma_0^{-(m-1)}\circ\tau\circ\delta_n\bigl(x_{m-1},y_{m-1}\bigr)=z_{m-1}.
\end{equation*}
But $\sigma_0^{-1}$, $\tau$, and $\delta_n$ are bijections, and a composition of bijections is itself a bijection.  Therefore, $\sigma_0^{-(m-1)}\circ\tau\circ\delta_n$ is a bijection from $\{0,1,\ldots,n-1\}^2$ to $\{0,1,\ldots,n^2-1\}$.  And by Definition~\ref{discrete}(a),
\begin{equation*}
\sigma_0^{-(m-1)}\circ\tau\circ\delta_n(0,0)=0.
\end{equation*}
This implies that $z_{m-1}=0$ if and only if $(x_{m-1},y_{m-1})=(0,0)$.  But
%
%
$(x_{m-1},y_{m-1})\linebreak[0]\ne(0,0)$ because $x$ or $y$ must have a base-$n$ length of $m$.  Therefore, $z_{m-1}\ne0$.  That is, $\len_{n^2}\bigl(f(x,y)\bigr)=m$.  We have now shown that
\begin{equation}\label{discrete-length}
\len_{n^2}\bigl(f(x,y)\bigr)=m
\end{equation}
for all $(x,y)\in\mathbb{N}^2$.  Hence, for each non-negative integer $m$, $f$ is a function from
\begin{equation*} C_m=\bigl\{\,(x,y)\in\mathbb{N}^2\;:\;\max\bigl(\len_n(x),\len_n(y)\bigr)=m\,\bigr\}
\end{equation*}
to
\begin{equation*}
D_m=\bigl\{\,z\in\mathbb{N}\;:\;\len_{n^2}(z)=m\,\bigr\}.
\end{equation*}
By similar reasoning, equations~\eqref{discrete-equations-inv} define a function $g\colon D_m\to C_m$ for each each non-negative integer $m$, where $g(0)=(0,0)$.  It is then straightforward to show that $g$ is the inverse of $f$.  Therefore, $f\colon C_m\to D_m$ is a bijection for each non-negative integer $m$.  But $C_0$, $C_1$, $C_2$, \ldots\ are disjoint sets whose union is  $\mathbb{N}^2$, and $D_0$, $D_1$, $D_2$, \ldots\ are disjoint sets whose union is $\mathbb{N}$.  Hence, $f$ is a bijection from $\mathbb{N}^2$ to $\mathbb{N}$.
%
%
%
That is, $f$ is a pairing function.

Now, for each non-negative integer $i$, $\len_n(i)\leq2\len_{n^2}(i)$ because
\begin{equation*}
\bigl\lceil\,\log_n(i+1)\bigr\rceil=\left\lceil\frac{\log_{n^2}(i+1)}{\log_{n^2}(n)}\right\rceil=\bigl\lceil2\log_{n^2}(i+1)\bigr\rceil\leq2\bigl\lceil\,\log_{n^2}(i+1)\bigr\rceil.
\end{equation*}
%
%
Therefore, by equation~\eqref{discrete-length},
\begin{equation*}
\len_n\bigl(f(x,y)\bigr)\leq2\len_{n^2}\bigl(f(x,y)\bigr)=2\max\bigl(\len_n(x),\len_n(y)\bigr)
\end{equation*}
for all $(x,y)\in\mathbb{N}^2$.  We may then conclude, by Corollary~\ref{alt-perfect}, that $f$ is a base-$n$ perfect pairing function.
\end{proof}

Several examples of base-$n$ discrete space-filling curves can be found in the published literature.\footnote{
Asano et al.~\cite{Asano1997} survey the base-$2$ discrete space-filling curves that often appear in the literature, but only discuss these curves on finite domains.  Chen et al.~\cite{Chen2007} extend the domain of the discrete Hilbert curve to $\mathbb{N}^2$.  Neither paper uses the term ``discrete space-filling curve''.  That term is used by Gotsman and Lindenbaum~\cite{Gotsman1996}, for example.
}  In particular, the discrete Peano curve (Figure~\ref{space-filling-illust}(a))
\begin{figure}
\centering
\begin{picture}(300,525)
%
%
\put(150,378){\makebox(0,0)[t]{\parbox{132pt}{\centering\small(a) The discrete Peano curve}}}
\put(84,388){\begin{picture}(132,137)(-21,-25)
%
%
\multiput(0,0)(0,12){9}{\multiput(0,0)(12,0){9}{\circle*{2}}}
%
\thinlines
\put(-12,-12){\line(1,0){114}}
\put(-12,-12){\line(0,1){114}}
%
\small
\put(105,-12){\makebox(0,0)[l]{$x$}}
\put(-12,105){\makebox(0,0)[b]{$y$}}
%
\thinlines
\multiput(0,-12)(12,0){9}{\line(0,1){3}}
\multiput(-12,0)(0,12){9}{\line(1,0){3}}
%
\small
\put(0,-18){\makebox(0,0)[t]{0}}
\put(12,-18){\makebox(0,0)[t]{1}}
\put(24,-18){\makebox(0,0)[t]{2}}
\put(36,-18){\makebox(0,0)[t]{3}}
\put(48,-18){\makebox(0,0)[t]{4}}
\put(60,-18){\makebox(0,0)[t]{5}}
\put(72,-18){\makebox(0,0)[t]{6}}
\put(84,-18){\makebox(0,0)[t]{7}}
\put(96,-18){\makebox(0,0)[t]{8}}
%
\small
\put(-16,0){\makebox(0,0)[r]{0}}
\put(-16,12){\makebox(0,0)[r]{1}}
\put(-16,24){\makebox(0,0)[r]{2}}
\put(-16,36){\makebox(0,0)[r]{3}}
\put(-16,48){\makebox(0,0)[r]{4}}
\put(-16,60){\makebox(0,0)[r]{5}}
\put(-16,72){\makebox(0,0)[r]{6}}
\put(-16,84){\makebox(0,0)[r]{7}}
\put(-16,96){\makebox(0,0)[r]{8}}
%
\thinlines
\put(0,0){\line(0,1){12}}
\put(0,12){\line(0,1){12}}
\put(0,24){\line(1,0){12}}
\put(12,24){\line(0,-1){12}}
\put(12,12){\line(0,-1){12}}
\put(12,0){\line(1,0){12}}
\put(24,0){\line(0,1){12}}
\put(24,12){\line(0,1){12}}
\put(24,24){\line(0,1){12}}
\put(24,36){\line(0,1){12}}
\put(24,48){\line(0,1){12}}
\put(24,60){\line(-1,0){12}}
\put(12,60){\line(0,-1){12}}
\put(12,48){\line(0,-1){12}}
\put(12,36){\line(-1,0){12}}
\put(0,36){\line(0,1){12}}
\put(0,48){\line(0,1){12}}
\put(0,60){\line(0,1){12}}
\put(0,72){\line(0,1){12}}
\put(0,84){\line(0,1){12}}
\put(0,96){\line(1,0){12}}
\put(12,96){\line(0,-1){12}}
\put(12,84){\line(0,-1){12}}
\put(12,72){\line(1,0){12}}
\put(24,72){\line(0,1){12}}
\put(24,84){\line(0,1){12}}
\put(24,96){\line(1,0){12}}
\put(36,96){\line(0,-1){12}}
\put(36,84){\line(0,-1){12}}
\put(36,72){\line(1,0){12}}
\put(48,72){\line(0,1){12}}
\put(48,84){\line(0,1){12}}
\put(48,96){\line(1,0){12}}
\put(60,96){\line(0,-1){12}}
\put(60,84){\line(0,-1){12}}
\put(60,72){\line(0,-1){12}}
\put(60,60){\line(0,-1){12}}
\put(60,48){\line(0,-1){12}}
\put(60,36){\line(-1,0){12}}
\put(48,36){\line(0,1){12}}
\put(48,48){\line(0,1){12}}
\put(48,60){\line(-1,0){12}}
\put(36,60){\line(0,-1){12}}
\put(36,48){\line(0,-1){12}}
\put(36,36){\line(0,-1){12}}
\put(36,24){\line(0,-1){12}}
\put(36,12){\line(0,-1){12}}
\put(36,0){\line(1,0){12}}
\put(48,0){\line(0,1){12}}
\put(48,12){\line(0,1){12}}
\put(48,24){\line(1,0){12}}
\put(60,24){\line(0,-1){12}}
\put(60,12){\line(0,-1){12}}
\put(60,0){\line(1,0){12}}
\put(72,0){\line(0,1){12}}
\put(72,12){\line(0,1){12}}
\put(72,24){\line(1,0){12}}
\put(84,24){\line(0,-1){12}}
\put(84,12){\line(0,-1){12}}
\put(84,0){\line(1,0){12}}
\put(96,0){\line(0,1){12}}
\put(96,12){\line(0,1){12}}
\put(96,24){\line(0,1){12}}
\put(96,36){\line(0,1){12}}
\put(96,48){\line(0,1){12}}
\put(96,60){\line(-1,0){12}}
\put(84,60){\line(0,-1){12}}
\put(84,48){\line(0,-1){12}}
\put(84,36){\line(-1,0){12}}
\put(72,36){\line(0,1){12}}
\put(72,48){\line(0,1){12}}
\put(72,60){\line(0,1){12}}
\put(72,72){\line(0,1){12}}
\put(72,84){\line(0,1){12}}
\put(72,96){\line(1,0){12}}
\put(84,96){\line(0,-1){12}}
\put(84,84){\line(0,-1){12}}
\put(84,72){\line(1,0){12}}
\put(96,72){\line(0,1){12}}
\put(96,84){\line(0,1){12}}
\end{picture}}
%
\put(60,207){\makebox(0,0)[t]{\parbox{120pt}{\centering\small(b) The discrete Hilbert curve}}}
\put(0,217){\begin{picture}(120,125)(-21,-25)
%
%
\multiput(0,0)(0,12){8}{\multiput(0,0)(12,0){8}{\circle*{2}}}
%
\thinlines
\put(-12,-12){\line(1,0){102}}
\put(-12,-12){\line(0,1){102}}
%
\small
\put(93,-12){\makebox(0,0)[l]{$x$}}
\put(-12,93){\makebox(0,0)[b]{$y$}}
%
\thinlines
\multiput(0,-12)(12,0){8}{\line(0,1){3}}
\multiput(-12,0)(0,12){8}{\line(1,0){3}}
%
\small
\put(0,-18){\makebox(0,0)[t]{0}}
\put(12,-18){\makebox(0,0)[t]{1}}
\put(24,-18){\makebox(0,0)[t]{2}}
\put(36,-18){\makebox(0,0)[t]{3}}
\put(48,-18){\makebox(0,0)[t]{4}}
\put(60,-18){\makebox(0,0)[t]{5}}
\put(72,-18){\makebox(0,0)[t]{6}}
\put(84,-18){\makebox(0,0)[t]{7}}
%
\small
\put(-16,0){\makebox(0,0)[r]{0}}
\put(-16,12){\makebox(0,0)[r]{1}}
\put(-16,24){\makebox(0,0)[r]{2}}
\put(-16,36){\makebox(0,0)[r]{3}}
\put(-16,48){\makebox(0,0)[r]{4}}
\put(-16,60){\makebox(0,0)[r]{5}}
\put(-16,72){\makebox(0,0)[r]{6}}
\put(-16,84){\makebox(0,0)[r]{7}}
%
\thinlines
\put(0,0){\line(0,1){12}}
\put(0,12){\line(1,0){12}}
\put(12,12){\line(0,-1){12}}
\put(12,0){\line(1,0){12}}
\put(24,0){\line(1,0){12}}
\put(36,0){\line(0,1){12}}
\put(36,12){\line(-1,0){12}}
\put(24,12){\line(0,1){12}}
\put(24,24){\line(1,0){12}}
\put(36,24){\line(0,1){12}}
\put(36,36){\line(-1,0){12}}
\put(24,36){\line(-1,0){12}}
\put(12,36){\line(0,-1){12}}
\put(12,24){\line(-1,0){12}}
\put(0,24){\line(0,1){12}}
\put(0,36){\line(0,1){12}}
\put(0,48){\line(1,0){12}}
\put(12,48){\line(0,1){12}}
\put(12,60){\line(-1,0){12}}
\put(0,60){\line(0,1){12}}
\put(0,72){\line(0,1){12}}
\put(0,84){\line(1,0){12}}
\put(12,84){\line(0,-1){12}}
\put(12,72){\line(1,0){12}}
\put(24,72){\line(0,1){12}}
\put(24,84){\line(1,0){12}}
\put(36,84){\line(0,-1){12}}
\put(36,72){\line(0,-1){12}}
\put(36,60){\line(-1,0){12}}
\put(24,60){\line(0,-1){12}}
\put(24,48){\line(1,0){12}}
\put(36,48){\line(1,0){12}}
\put(48,48){\line(1,0){12}}
\put(60,48){\line(0,1){12}}
\put(60,60){\line(-1,0){12}}
\put(48,60){\line(0,1){12}}
\put(48,72){\line(0,1){12}}
\put(48,84){\line(1,0){12}}
\put(60,84){\line(0,-1){12}}
\put(60,72){\line(1,0){12}}
\put(72,72){\line(0,1){12}}
\put(72,84){\line(1,0){12}}
\put(84,84){\line(0,-1){12}}
\put(84,72){\line(0,-1){12}}
\put(84,60){\line(-1,0){12}}
\put(72,60){\line(0,-1){12}}
\put(72,48){\line(1,0){12}}
\put(84,48){\line(0,-1){12}}
\put(84,36){\line(0,-1){12}}
\put(84,24){\line(-1,0){12}}
\put(72,24){\line(0,1){12}}
\put(72,36){\line(-1,0){12}}
\put(60,36){\line(-1,0){12}}
\put(48,36){\line(0,-1){12}}
\put(48,24){\line(1,0){12}}
\put(60,24){\line(0,-1){12}}
\put(60,12){\line(-1,0){12}}
\put(48,12){\line(0,-1){12}}
\put(48,0){\line(1,0){12}}
\put(60,0){\line(1,0){12}}
\put(72,0){\line(0,1){12}}
\put(72,12){\line(1,0){12}}
\put(84,12){\line(0,-1){12}}
\end{picture}}
%
\put(240,207){\makebox(0,0)[t]{\parbox{120pt}{\centering\small(c) The z-order}}}
\put(180,217){\begin{picture}(120,125)(-21,-25)
%
%
\multiput(0,0)(0,12){8}{\multiput(0,0)(12,0){8}{\circle*{2}}}
%
\thinlines
\put(-12,-12){\line(1,0){102}}
\put(-12,-12){\line(0,1){102}}
%
\small
\put(93,-12){\makebox(0,0)[l]{$x$}}
\put(-12,93){\makebox(0,0)[b]{$y$}}
%
\thinlines
\multiput(0,-12)(12,0){8}{\line(0,1){3}}
\multiput(-12,0)(0,12){8}{\line(1,0){3}}
%
\small
\put(0,-18){\makebox(0,0)[t]{0}}
\put(12,-18){\makebox(0,0)[t]{1}}
\put(24,-18){\makebox(0,0)[t]{2}}
\put(36,-18){\makebox(0,0)[t]{3}}
\put(48,-18){\makebox(0,0)[t]{4}}
\put(60,-18){\makebox(0,0)[t]{5}}
\put(72,-18){\makebox(0,0)[t]{6}}
\put(84,-18){\makebox(0,0)[t]{7}}
%
\small
\put(-16,0){\makebox(0,0)[r]{0}}
\put(-16,12){\makebox(0,0)[r]{1}}
\put(-16,24){\makebox(0,0)[r]{2}}
\put(-16,36){\makebox(0,0)[r]{3}}
\put(-16,48){\makebox(0,0)[r]{4}}
\put(-16,60){\makebox(0,0)[r]{5}}
\put(-16,72){\makebox(0,0)[r]{6}}
\put(-16,84){\makebox(0,0)[r]{7}}
%
\thinlines
\put(0,0){\line(0,1){12}}
\put(0,12){\line(1,-1){12}}
\put(12,0){\line(0,1){12}}
\put(12,12){\line(-1,1){12}}
\put(0,24){\line(0,1){12}}
\put(0,36){\line(1,-1){12}}
\put(12,24){\line(0,1){12}}
\put(12,36){\line(1,-3){12}}
\put(24,0){\line(0,1){12}}
\put(24,12){\line(1,-1){12}}
\put(36,0){\line(0,1){12}}
\put(36,12){\line(-1,1){12}}
\put(24,24){\line(0,1){12}}
\put(24,36){\line(1,-1){12}}
\put(36,24){\line(0,1){12}}
\put(36,36){\line(-3,1){36}}
\put(0,48){\line(0,1){12}}
\put(0,60){\line(1,-1){12}}
\put(12,48){\line(0,1){12}}
\put(12,60){\line(-1,1){12}}
\put(0,72){\line(0,1){12}}
\put(0,84){\line(1,-1){12}}
\put(12,72){\line(0,1){12}}
\put(12,84){\line(1,-3){12}}
\put(24,48){\line(0,1){12}}
\put(24,60){\line(1,-1){12}}
\put(36,48){\line(0,1){12}}
\put(36,60){\line(-1,1){12}}
\put(24,72){\line(0,1){12}}
\put(24,84){\line(1,-1){12}}
\put(36,72){\line(0,1){12}}
\qbezier(36,84)(42,42)(48,0)
\put(48,0){\line(0,1){12}}
\put(48,12){\line(1,-1){12}}
\put(60,0){\line(0,1){12}}
\put(60,12){\line(-1,1){12}}
\put(48,24){\line(0,1){12}}
\put(48,36){\line(1,-1){12}}
\put(60,24){\line(0,1){12}}
\put(60,36){\line(1,-3){12}}
\put(72,0){\line(0,1){12}}
\put(72,12){\line(1,-1){12}}
\put(84,0){\line(0,1){12}}
\put(84,12){\line(-1,1){12}}
\put(72,24){\line(0,1){12}}
\put(72,36){\line(1,-1){12}}
\put(84,24){\line(0,1){12}}
\put(84,36){\line(-3,1){36}}
\put(48,48){\line(0,1){12}}
\put(48,60){\line(1,-1){12}}
\put(60,48){\line(0,1){12}}
\put(60,60){\line(-1,1){12}}
\put(48,72){\line(0,1){12}}
\put(48,84){\line(1,-1){12}}
\put(60,72){\line(0,1){12}}
\put(60,84){\line(1,-3){12}}
\put(72,48){\line(0,1){12}}
\put(72,60){\line(1,-1){12}}
\put(84,48){\line(0,1){12}}
\put(84,60){\line(-1,1){12}}
\put(72,72){\line(0,1){12}}
\put(72,84){\line(1,-1){12}}
\put(84,72){\line(0,1){12}}
\end{picture}}
%
\put(60,36){\makebox(0,0)[t]{\parbox{120pt}{\centering\small(d) The Gray-coded curve}}}
\put(0,46){\begin{picture}(120,125)(-21,-25)
%
%
\multiput(0,0)(0,12){8}{\multiput(0,0)(12,0){8}{\circle*{2}}}
%
\thinlines
\put(-12,-12){\line(1,0){102}}
\put(-12,-12){\line(0,1){102}}
%
\small
\put(93,-12){\makebox(0,0)[l]{$x$}}
\put(-12,93){\makebox(0,0)[b]{$y$}}
%
\thinlines
\multiput(0,-12)(12,0){8}{\line(0,1){3}}
\multiput(-12,0)(0,12){8}{\line(1,0){3}}
%
\small
\put(0,-18){\makebox(0,0)[t]{0}}
\put(12,-18){\makebox(0,0)[t]{1}}
\put(24,-18){\makebox(0,0)[t]{2}}
\put(36,-18){\makebox(0,0)[t]{3}}
\put(48,-18){\makebox(0,0)[t]{4}}
\put(60,-18){\makebox(0,0)[t]{5}}
\put(72,-18){\makebox(0,0)[t]{6}}
\put(84,-18){\makebox(0,0)[t]{7}}
%
\small
\put(-16,0){\makebox(0,0)[r]{0}}
\put(-16,12){\makebox(0,0)[r]{1}}
\put(-16,24){\makebox(0,0)[r]{2}}
\put(-16,36){\makebox(0,0)[r]{3}}
\put(-16,48){\makebox(0,0)[r]{4}}
\put(-16,60){\makebox(0,0)[r]{5}}
\put(-16,72){\makebox(0,0)[r]{6}}
\put(-16,84){\makebox(0,0)[r]{7}}
%
\thinlines
\put(0,0){\line(0,1){12}}
\put(0,12){\line(1,0){12}}
\put(12,12){\line(0,-1){12}}
\qbezier(12,0)(20,18)(12,36)
\put(12,36){\line(0,-1){12}}
\put(12,24){\line(-1,0){12}}
\put(0,24){\line(0,1){12}}
\qbezier(0,36)(18,44)(36,36)
\put(36,36){\line(0,-1){12}}
\put(36,24){\line(-1,0){12}}
\put(24,24){\line(0,1){12}}
\qbezier(24,36)(16,18)(24,0)
\put(24,0){\line(0,1){12}}
\put(24,12){\line(1,0){12}}
\put(36,12){\line(0,-1){12}}
\qbezier(36,0)(44,42)(36,84)
\put(36,84){\line(0,-1){12}}
\put(36,72){\line(-1,0){12}}
\put(24,72){\line(0,1){12}}
\qbezier(24,84)(16,66)(24,48)
\put(24,48){\line(0,1){12}}
\put(24,60){\line(1,0){12}}
\put(36,60){\line(0,-1){12}}
\qbezier(36,48)(18,40)(0,48)
\put(0,48){\line(0,1){12}}
\put(0,60){\line(1,0){12}}
\put(12,60){\line(0,-1){12}}
\qbezier(12,48)(20,66)(12,84)
\put(12,84){\line(0,-1){12}}
\put(12,72){\line(-1,0){12}}
\put(0,72){\line(0,1){12}}
\qbezier(0,84)(42,92)(84,84)
\put(84,84){\line(0,-1){12}}
\put(84,72){\line(-1,0){12}}
\put(72,72){\line(0,1){12}}
\qbezier(72,84)(64,66)(72,48)
\put(72,48){\line(0,1){12}}
\put(72,60){\line(1,0){12}}
\put(84,60){\line(0,-1){12}}
\qbezier(84,48)(66,40)(48,48)
\put(48,48){\line(0,1){12}}
\put(48,60){\line(1,0){12}}
\put(60,60){\line(0,-1){12}}
\qbezier(60,48)(68,66)(60,84)
\put(60,84){\line(0,-1){12}}
\put(60,72){\line(-1,0){12}}
\put(48,72){\line(0,1){12}}
\qbezier(48,84)(40,42)(48,0)
\put(48,0){\line(0,1){12}}
\put(48,12){\line(1,0){12}}
\put(60,12){\line(0,-1){12}}
\qbezier(60,0)(68,18)(60,36)
\put(60,36){\line(0,-1){12}}
\put(60,24){\line(-1,0){12}}
\put(48,24){\line(0,1){12}}
\qbezier(48,36)(66,44)(84,36)
\put(84,36){\line(0,-1){12}}
\put(84,24){\line(-1,0){12}}
\put(72,24){\line(0,1){12}}
\qbezier(72,36)(64,18)(72,0)
\put(72,0){\line(0,1){12}}
\put(72,12){\line(1,0){12}}
\put(84,12){\line(0,-1){12}}
\end{picture}}
%
\put(240,36){\makebox(0,0)[t]{\parbox{120pt}{\centering\small(e) A base-$2$ discrete space-filling curve}}}
\put(180,46){\begin{picture}(120,125)(-21,-25)
%
%
\multiput(0,0)(0,12){8}{\multiput(0,0)(12,0){8}{\circle*{2}}}
%
\thinlines
\put(-12,-12){\line(1,0){102}}
\put(-12,-12){\line(0,1){102}}
%
\small
\put(93,-12){\makebox(0,0)[l]{$x$}}
\put(-12,93){\makebox(0,0)[b]{$y$}}
%
\thinlines
\multiput(0,-12)(12,0){8}{\line(0,1){3}}
\multiput(-12,0)(0,12){8}{\line(1,0){3}}
%
\small
\put(0,-18){\makebox(0,0)[t]{0}}
\put(12,-18){\makebox(0,0)[t]{1}}
\put(24,-18){\makebox(0,0)[t]{2}}
\put(36,-18){\makebox(0,0)[t]{3}}
\put(48,-18){\makebox(0,0)[t]{4}}
\put(60,-18){\makebox(0,0)[t]{5}}
\put(72,-18){\makebox(0,0)[t]{6}}
\put(84,-18){\makebox(0,0)[t]{7}}
%
\small
\put(-16,0){\makebox(0,0)[r]{0}}
\put(-16,12){\makebox(0,0)[r]{1}}
\put(-16,24){\makebox(0,0)[r]{2}}
\put(-16,36){\makebox(0,0)[r]{3}}
\put(-16,48){\makebox(0,0)[r]{4}}
\put(-16,60){\makebox(0,0)[r]{5}}
\put(-16,72){\makebox(0,0)[r]{6}}
\put(-16,84){\makebox(0,0)[r]{7}}
%
\thinlines
\put(0,0){\line(0,1){12}}
\put(0,12){\line(1,-1){12}}
\put(12,0){\line(0,1){12}}
\put(12,12){\line(1,1){12}}
\put(24,24){\line(1,1){12}}
\put(36,36){\line(0,-1){12}}
\put(36,24){\line(-1,1){12}}
\put(24,36){\line(-1,0){12}}
\put(12,36){\line(-1,-1){12}}
\put(0,24){\line(0,1){12}}
\put(0,36){\line(1,-1){12}}
\put(12,24){\line(1,-1){12}}
\put(24,12){\line(0,-1){12}}
\put(24,0){\line(1,1){12}}
\put(36,12){\line(0,-1){12}}
\put(36,0){\line(1,0){12}}
\put(48,0){\line(0,1){12}}
\put(48,12){\line(1,0){12}}
\put(60,12){\line(0,-1){12}}
\put(60,0){\line(1,0){12}}
\put(72,0){\line(1,0){12}}
\put(84,0){\line(0,1){12}}
\put(84,12){\line(-1,0){12}}
\put(72,12){\line(-1,1){12}}
\put(60,24){\line(-1,0){12}}
\put(48,24){\line(0,1){12}}
\put(48,36){\line(1,0){12}}
\put(60,36){\line(1,0){12}}
\put(72,36){\line(0,-1){12}}
\put(72,24){\line(1,0){12}}
\put(84,24){\line(0,1){12}}
\put(84,36){\line(0,1){12}}
\put(84,48){\line(0,1){12}}
\put(84,60){\line(-1,0){12}}
\put(72,60){\line(0,-1){12}}
\put(72,48){\line(-1,0){12}}
\put(60,48){\line(-1,0){12}}
\put(48,48){\line(0,1){12}}
\put(48,60){\line(1,0){12}}
\put(60,60){\line(1,1){12}}
\put(72,72){\line(1,0){12}}
\put(84,72){\line(0,1){12}}
\put(84,84){\line(-1,0){12}}
\put(72,84){\line(-1,0){12}}
\put(60,84){\line(0,-1){12}}
\put(60,72){\line(-1,0){12}}
\put(48,72){\line(0,1){12}}
\put(48,84){\line(-1,0){12}}
\put(36,84){\line(0,-1){12}}
\put(36,72){\line(-1,1){12}}
\put(24,84){\line(0,-1){12}}
\put(24,72){\line(-1,-1){12}}
\put(12,60){\line(-1,-1){12}}
\put(0,48){\line(0,1){12}}
\put(0,60){\line(1,-1){12}}
\put(12,48){\line(1,0){12}}
\put(24,48){\line(1,1){12}}
\put(36,60){\line(0,-1){12}}
\put(36,48){\line(-1,1){12}}
\put(24,60){\line(-1,1){12}}
\put(12,72){\line(0,1){12}}
\put(12,84){\line(-1,-1){12}}
\put(0,72){\line(0,1){12}}
\end{picture}}
\end{picture}
\caption{Base-$n$ discrete space-filling curves.  For each pairing function $f\colon\mathbb{N}^2\to\mathbb{N}$, points that occur consecutively in the sequence $f^{-1}(0)$, $f^{-1}(1)$, $f^{-1}(2)$, \ldots\ are joined with a line segment or arc.}
\label{space-filling-illust}
\end{figure}
is the base-$3$ discrete space-filling curve that is determined by the permutations
\begin{gather*}
\tau=\left(\begin{matrix}
0 & 1 & 2 & 3 & 4 & 5 & 6 & 7 & 8\\
0 & 1 & 2 & 5 & 4 & 3 & 6 & 7 & 8
\end{matrix}\right),\qquad\sigma_4=\left(\begin{matrix}
0 & 1 & 2 & 3 & 4 & 5 & 6 & 7 & 8\\
8 & 7 & 6 & 5 & 4 & 3 & 2 & 1 & 0
\end{matrix}\right),\\
\sigma_0=\sigma_2=\sigma_6=\sigma_8=I,\qquad\sigma_1=\sigma_7=\left(\begin{matrix}
0 & 1 & 2 & 3 & 4 & 5 & 6 & 7 & 8\\
6 & 7 & 8 & 3 & 4 & 5 & 0 & 1 & 2
\end{matrix}\right),\\
\sigma_3=\sigma_5=\left(\begin{matrix}
0 & 1 & 2 & 3 & 4 & 5 & 6 & 7 & 8\\
2 & 1 & 0 & 5 & 4 & 3 & 8 & 7 & 6
\end{matrix}\right).
\end{gather*}
And the \emph{discrete Hilbert curve} (Figure~\ref{space-filling-illust}(b)) is the base-$2$ discrete space-filling curve defined by the permutations $\sigma_1=\sigma_2=I$ and
\begin{equation*}
\tau=\left(\begin{matrix}
0 & 1 & 2 & 3\\
0 & 1 & 3 & 2
\end{matrix}\right),\qquad\sigma_0=\left(\begin{matrix}
0 & 1 & 2 & 3\\
0 & 3 & 2 & 1
\end{matrix}\right),\qquad\sigma_3=\left(\begin{matrix}
0 & 1 & 2 & 3\\
2 & 1 & 0 & 3
\end{matrix}\right).
\end{equation*}
If $F\colon[0,1]\to[0,1]^2$ is the continuous Hilbert curve, then the discrete Hilbert curve is the unique pairing function $f\colon\mathbb{N}^2\to\mathbb{N}$ such that
\begin{equation*}
F\left(\frac{z+1/2}{4^{2k+1}}\right)=2^{-(2k+1)}\Bigl(f^{-1}(z)+\bigl(\tfrac{1}{2},\tfrac{1}{2}\bigl)\Bigr)
\end{equation*}
for all non-negative integers $z$ and $k$ such that $z<4^{2k+1}$.
%
%
Other well-known examples of base-$2$ discrete space-filling curves include the \emph{z-order} (Figure~\ref{space-filling-illust}(c)), which is defined by the permutations $\tau=\sigma_0=\sigma_1=\sigma_2=\sigma_3=I$, and the \emph{Gray-coded curve} (Figure~\ref{space-filling-illust}(d)), given by the permutations
\begin{equation*}
\tau=\left(\begin{matrix}
0 & 1 & 2 & 3\\
0 & 1 & 3 & 2
\end{matrix}\right),\qquad\sigma_0=\sigma_3=I,\qquad\sigma_1=\sigma_2=\left(\begin{matrix}
0 & 1 & 2 & 3\\
2 & 3 & 0 & 1
\end{matrix}\right).
\end{equation*}

In all of these examples, and for each $i\in\{0,1,\ldots,n^2-1\}$, the function $\delta_n^{-1}\circ\tau^{-1}\circ\sigma_i^{-1}\circ\tau\circ\delta_n$ happens to be a geometric isometry from $\{0,1,\ldots,n-1\}^2$ to
%
%
$\{0,1,\ldots,\linebreak[0]n-1\}^2$.
%
%
These isometries are closely related to the similarity transformations~\cite{Sagan1994} that are often used to construct space-filling curves.  But Definition~\ref{discrete} allows any permutations $\tau$, $\sigma_0$, $\sigma_1$, \ldots, $\sigma_{n^2-1}$ in $S_n$ that satisfy the equations $\tau(0)=0$ and $\sigma_0(0)=0$ to determine a base-$n$ discrete space-filling curve, including those permutations where $\delta_n^{-1}\circ\tau^{-1}\circ\sigma_i^{-1}\circ\tau\circ\delta_n$ is not an isometry.  For example, a base-$2$ discrete space-filling curve where $\delta_n^{-1}\circ\tau^{-1}\circ\sigma_i^{-1}\circ\tau\circ\delta_n$ is not an isometry for any $i\in\{0,1,2,3\}$ is defined by the permutations
\begin{gather*}
\tau=I,\qquad\sigma_0=\left(\begin{matrix}
0 & 1 & 2 & 3\\
0 & 3 & 1 & 2
\end{matrix}\right),\qquad\sigma_1=\left(\begin{matrix}
0 & 1 & 2 & 3\\
0 & 1 & 3 & 2
\end{matrix}\right),\\
\qquad\sigma_2=\left(\begin{matrix}
0 & 1 & 2 & 3\\
1 & 0 & 2 & 3
\end{matrix}\right),\qquad\sigma_3=\left(\begin{matrix}
0 & 1 & 2 & 3\\
1 & 2 & 0 & 3
\end{matrix}\right).
\end{gather*}
This function is illustrated in Figure~\ref{space-filling-illust}(e).

The definition of a base-$n$ discrete space-filling curve can also be extended to higher dimensions.  In three dimensions we have the following.

\begin{definition}
Let $n$ be any integer greater than $1$.  A function $f\colon\mathbb{N}^3\to\mathbb{N}$ is said to be a \emph{$3$-dimensional base-$n$ discrete space-filling curve} if and only if there exist permutations $\tau$, $\sigma_0$, $\sigma_1$, \ldots, $\sigma_{n^3-1}$ in $S_{n^3}$ such that the following three conditions hold.
\begin{enumerate}
\item[(a)] $\tau(0)=0$ and $\sigma_0(0)=0$.
\item[(b)] $f(0,0,0)=0$.
\item[(c)] For all non-negative integers $w$, $x$, $y$, and $z$ such that $(w,x,y)\ne(0,0,0)$, $f(w,x,y)=z$ if and only if
\begin{equation*}
\begin{aligned}
\sigma_0^{-(m-1)}\circ\tau\bigl(n^2w_{m-1}+nx_{m-1}+y_{m-1}\bigr)&=z_{m-1},\\
\sigma_{z_{m-1}}\circ\sigma_0^{-(m-1)}\circ\tau\bigl(n^2w_{m-2}+nx_{m-2}+y_{m-2}\bigr)&=z_{m-2},\\
\sigma_{z_{m-2}}\circ\sigma_{z_{m-1}}\circ\sigma_0^{-(m-1)}\circ\tau\bigl(n^2w_{m-3}+nx_{m-3}+y_{m-3}\bigr)&=z_{m-3},\\
&\;\;\vdots\\
\sigma_{z_1}\circ\cdots\circ\sigma_{z_{m-2}}\circ\sigma_{z_{m-1}}\circ\sigma_0^{-(m-1)}\circ\tau\bigl(n^2w_0+nx_0+y_0\bigr)&=z_0,
\end{aligned}
\end{equation*}
where $m=\max\bigl(\len_n(w),\len_n(x),\len_n(y)\bigr)$, where
\begin{align*}
w&=w_{m-1}n^{m-1}+w_{m-2}n^{m-2}+\cdots+w_0n^0,\\
x&=x_{m-1}n^{m-1}+x_{m-2}n^{m-2}+\cdots+x_0n^0,\\
y&=y_{m-1}n^{m-1}+y_{m-2}n^{m-2}+\cdots+y_0n^0
\end{align*}
are base-$n$ expansions of $w$, $x$ and $y$, and where
\begin{equation*}
z=z_{m-1}\bigl(n^3\bigr)^{m-1}+z_{m-2}\bigl(n^3\bigr)^{m-2}+\cdots+z_0\bigl(n^3\bigr)^0
\end{equation*}
is a base-$n^3$ expansion of $z$.
\end{enumerate}
\end{definition}

Generalizing the proof of Theorem~\ref{discrete-perfect}, it can be shown that every $3$-dimensional base-$n$ discrete space-filling curve, where $n$ is an integer greater than $1$, is a base-$n$ perfect $3$-tupling function.  Many examples of these functions appear in the published literature.  Most notably, Alber and Niedermeier~\cite{Alber2000} have identified $3!\cdot1536=9216$ distinct $3$-dimensional base-$2$ discrete space-filling curves that generalize the discrete Hilbert curve.  One such generalization is given by the permutations
\begin{gather*}
\tau=\left(\begin{matrix}
0 & 1 & 2 & 3 & 4 & 5 & 6 & 7\\
0 & 1 & 3 & 2 & 7 & 6 & 4 & 5
\end{matrix}\right),\qquad\sigma_1=\sigma_2=\left(\begin{matrix}
0 & 1 & 2 & 3 & 4 & 5 & 6 & 7\\
0 & 1 & 6 & 7 & 4 & 5 & 2 & 3
\end{matrix}\right),\\
\sigma_0=\left(\begin{matrix}
0 & 1 & 2 & 3 & 4 & 5 & 6 & 7\\
0 & 7 & 4 & 3 & 2 & 5 & 6 & 1
\end{matrix}\right),\qquad\sigma_3=\sigma_4=\left(\begin{matrix}
0 & 1 & 2 & 3 & 4 & 5 & 6 & 7\\
2 & 3 & 0 & 1 & 6 & 7 & 4 & 5
\end{matrix}\right),\\
\sigma_5=\sigma_6=\left(\begin{matrix}
0 & 1 & 2 & 3 & 4 & 5 & 6 & 7\\
4 & 5 & 2 & 3 & 0 & 1 & 6 & 7
\end{matrix}\right),\qquad\sigma_7=\left(\begin{matrix}
0 & 1 & 2 & 3 & 4 & 5 & 6 & 7\\
6 & 1 & 2 & 5 & 4 & 3 & 0 & 7
\end{matrix}\right).
\end{gather*}

\section{A Technique For Constructing Pairing Functions}\label{technique}

In this section we describe a technique for constructing a pairing function from a non-decreasing unbounded function.  Let $g$ be any function from $\mathbb{N}$ to $\mathbb{N}$.  By definition, $g$ is \emph{non-decreasing} if and only if
\begin{equation*}
\forall x,y\in\mathbb{N}\,\bigl(\,x\leq y\quad\text{implies}\quad g(x)\leq g(y)\,\bigr).
\end{equation*}
And $g$ is \emph{unbounded} if and only if
\begin{equation*}
\forall y\in\mathbb{N}\;\;\exists x\in\mathbb{N}\,\bigl(\,g(x)\geq y\,\bigr).
\end{equation*}
If $g$ is unbounded, we define $g^+\colon\mathbb{N}\to\mathbb{N}$ so that $g^+(y)=x$ if and only if
\begin{equation}\label{pseudo-inv}
g(x)\geq y\quad\And\quad\forall u\in\mathbb{N}\,\bigl(\,g(u)\geq y\quad\text{implies}\quad u\geq x\,\bigr).
\end{equation}
That is, we define $g^+(y)$ to be the smallest $x\in\mathbb{N}$ such that $g(x)\geq y$.  If $g$ is non-decreasing and unbounded, this definition implies that $g^+$ is a pseudo-inverse of $g$, in the sense that $g\bigl(g^+(g(x))\bigr)=g(x)$ for all $x\in\mathbb{N}$.
%
%
%

\begin{lemma}\label{unbounded}
Let $g\colon\mathbb{N}\to\mathbb{N}$ be any unbounded function.  Then $g^+$ is unbounded.
\end{lemma}
\begin{proof}
Consider any $y\in\mathbb{N}$.  Let $m=\max\bigl(g(0),g(1),\ldots,g(y)\bigr)$ and let $z=g^+(m+1)$.  By definition, $z$ is the smallest non-negative integer such that $g(z)\geq m+1$.  But $g(0)$, $g(1)$, \ldots, $g(y)$ are all strictly less than $m+1$ because
\begin{equation*}
\max\bigl(g(0),g(1),\ldots,g(y)\bigr)=m<m+1.
\end{equation*}
Hence, it must be the case that $z>y$.  That is, $g^+(m+1)=z>y$.  We have shown that for every $y\in\mathbb{N}$ there exists an $x\in\mathbb{N}$, namely $x=m+1$, such that $g^+(x)\geq y$.
\end{proof}

Given any unbounded function $g\colon\mathbb{N}\to\mathbb{N}$, let $S=\bigl\{g^+(y):y\in\mathbb{N}\bigr\}$ denote the range of $g^+$.  We say that $S$ is the set of \emph{step points} of $g$.  Conventionally, we order the members of $S$ in the sequence $s_0$, $s_1$, $s_2$, \ldots\ where
\begin{equation*}
s_0<s_1<s_2<\cdots.
\end{equation*}
We say that this is the \emph{sequence of step points} of $g$.  This definition implies that
\begin{equation}\label{s0}
s_0=0,
\end{equation}
since $g^+(0)=0$ for every unbounded function $g\colon\mathbb{N}\to\mathbb{N}$.  And by Lemma~\ref{unbounded}, each unbounded function $g\colon\mathbb{N}\to\mathbb{N}$ has infinitely many step points.  

\begin{lemma}\label{step-points}
Let $g\colon\mathbb{N}\to\mathbb{N}$ be any non-decreasing unbounded function, and let $s_0$, $s_1$, $s_2$, \ldots\ be the sequence of step points of $g$.  Then,
\begin{equation*}
\forall x,k\in\mathbb{N}\,\bigl(\,x<s_{k+1}\quad\text{implies}\quad g(x)\leq g(s_k)\,\bigr).
\end{equation*}
\end{lemma}
\begin{proof}
Consider any non-negative integers $x$ and $k$.  We will prove the contrapositive that
\begin{equation*}
g(x)> g(s_k)\quad\text{implies}\quad x\geq s_{k+1}.
\end{equation*}
So, suppose that
\begin{equation}\label{step-inequality}
g(x)> g(s_k)
\end{equation}
and let $g^+\bigl(g(x)\bigr)=z$.  Because $z$ is in the range of $g^+$, $z$ must be a step point of $g$.  Furthermore, by condition~\eqref{pseudo-inv},
\begin{equation}\label{step-pseudo-inv}
g(z)\geq g(x)\quad\And\quad\forall u\in\mathbb{N}\,\bigl(\,g(u)\geq g(x)\quad\text{implies}\quad u\geq z\,\bigr).
\end{equation}
Combining conditions~\eqref{step-pseudo-inv} and~\eqref{step-inequality},
\begin{equation*}
g(z)\geq g(x)>g(s_k).
\end{equation*}
But $g$ is non-decreasing, so this implies that $z>s_k$.  We have shown that the step point $z$ is larger than $s_k$.  Therefore, it must be the case that
\begin{equation*}
z\geq s_{k+1}.
\end{equation*}
But substituting $u=x$ into condition~\eqref{step-pseudo-inv}, we have
\begin{equation*}
g(x)\geq g(x)\quad\text{implies}\quad x\geq z.
\end{equation*}
Therefore, $x\geq z\geq s_{k+1}$.
\end{proof}

\begin{definition}\label{a-b-def}
Let $g\colon\mathbb{N}\to\mathbb{N}$ be any unbounded function, and let $s_0$, $s_1$, $s_2$, \ldots\ be the sequence of step points of $g$.  For each $k\in\mathbb{N}$, define
\begin{equation*}
A_k=\bigl\{\,(x,y)\in\mathbb{N}^2\;:\;x<s_{k+1}\;\;\And\;\;y\leq g(s_k)\,\bigr\}
\end{equation*}
and
\begin{equation*}
B_k=\bigl\{0,1,2,\ldots,s_{k+1}\bigl(g(s_k)+1\bigr)-1\bigr\}.
\end{equation*}
\end{definition}

\begin{definition}\label{phi-def}
Given any non-decreasing unbounded function $g\colon\mathbb{N}\to\mathbb{N}$, let $\phi_g\colon\mathbb{N}^2\to\mathbb{N}$ be defined\footnote{
Note that $y>g(x)$ can be replaced with $x<g^+(y)$ in this definition, since
\begin{equation*}
y>g(x)\quad\text{if and only if}\quad x<g^+(y).
\end{equation*}
%
%
} so that for all non-negative integers $x$ and $y$,
\begin{equation*}
\phi_g(x,y)=\begin{cases}
y\cdot g^+(y)+x &\text{\quad if $y>g(x)$}\\
x\bigl(g(x)+1\bigr)+y &\text{\quad otherwise}\rule{0pt}{15pt}
\end{cases}.
\end{equation*}
\end{definition}

We will prove that $\phi_g(x,y)$ is a pairing function in Theorem~\ref{phi-pairing}.  The function is illustrated in Figure~\ref{phi-illust}.
\begin{figure}
\centering
\begin{picture}(243,184)(-25,-29)
%
%
\multiput(0,0)(0,32){5}{\multiput(0,0)(32,0){7}{\circle*{4}}}
%
\thinlines
\put(-16,-16){\line(1,0){224}}
\put(-16,-16){\line(0,1){160}}
%
\put(212,-16){\makebox(0,0)[l]{$x$}}
\put(-16,148){\makebox(0,0)[b]{$y$}}
%
\thinlines
\multiput(0,-16)(32,0){7}{\line(0,1){3}}
\multiput(-16,0)(0,32){5}{\line(1,0){3}}
%
\put(0,-22){\makebox(0,0)[t]{0}}
\put(32,-22){\makebox(0,0)[t]{1}}
\put(64,-22){\makebox(0,0)[t]{2}}
\put(96,-22){\makebox(0,0)[t]{3}}
\put(128,-22){\makebox(0,0)[t]{4}}
\put(160,-22){\makebox(0,0)[t]{5}}
\put(192,-22){\makebox(0,0)[t]{6}}
%
\put(-20,0){\makebox(0,0)[r]{0}}
\put(-20,32){\makebox(0,0)[r]{1}}
\put(-20,64){\makebox(0,0)[r]{2}}
\put(-20,96){\makebox(0,0)[r]{3}}
\put(-20,128){\makebox(0,0)[r]{4}}
%
\multiput(3.816,3.816)(1.875,1.875){14}{\circle*{1}}
\multiput(37.396,32.000)(2.651,0.000){9}{\circle*{1}}
\multiput(69.396,32.000)(2.651,0.000){9}{\circle*{1}}
\multiput(98.366,36.732)(1.186,2.371){24}{\circle*{1}}
\multiput(131.816,99.816)(1.875,1.875){14}{\circle*{1}}
\multiput(165.396,128.000)(2.651,0.000){9}{\circle*{1}}
%
\thicklines
\put(32,5){\line(0,1){22}}
\put(64,5){\line(0,1){22}}
\put(96,5){\line(0,1){22}}
\put(5,64){\line(1,0){22}}
\put(37,64){\line(1,0){22}}
\put(69,64){\line(1,0){22}}
\put(5,96){\line(1,0){22}}
\put(37,96){\line(1,0){22}}
\put(69,96){\line(1,0){22}}
\put(128,5){\line(0,1){22}}
\put(128,37){\line(0,1){22}}
\put(128,69){\line(0,1){22}}
\put(5,128){\line(1,0){22}}
\put(37,128){\line(1,0){22}}
\put(69,128){\line(1,0){22}}
\put(101,128){\line(1,0){22}}
\put(160,5){\line(0,1){22}}
\put(160,37){\line(0,1){22}}
\put(160,69){\line(0,1){22}}
\put(160,101){\line(0,1){22}}
\put(192,5){\line(0,1){22}}
\put(192,37){\line(0,1){22}}
\put(192,69){\line(0,1){22}}
\put(192,101){\line(0,1){22}}
%
\put(6,0){\makebox(0,0)[l]{0}}
\put(0,38){\makebox(0,0)[b]{1}}
\put(38,0){\makebox(0,0)[l]{2}}
\put(32,38){\makebox(0,0)[b]{3}}
\put(70,0){\makebox(0,0)[l]{4}}
\put(64,38){\makebox(0,0)[b]{5}}
\put(102,0){\makebox(0,0)[l]{6}}
\put(102,32){\makebox(0,0)[l]{7}}
\put(0,70){\makebox(0,0)[b]{8}}
\put(32,70){\makebox(0,0)[b]{9}}
\put(64,70){\makebox(0,0)[b]{10}}
\put(96,70){\makebox(0,0)[b]{11}}
\put(0,102){\makebox(0,0)[b]{12}}
\put(32,102){\makebox(0,0)[b]{13}}
\put(64,102){\makebox(0,0)[b]{14}}
\put(96,102){\makebox(0,0)[b]{15}}
\put(134,0){\makebox(0,0)[l]{16}}
\put(134,32){\makebox(0,0)[l]{17}}
\put(134,64){\makebox(0,0)[l]{18}}
\put(134,96){\makebox(0,0)[l]{19}}
\put(0,134){\makebox(0,0)[b]{20}}
\put(32,134){\makebox(0,0)[b]{21}}
\put(64,134){\makebox(0,0)[b]{22}}
\put(96,134){\makebox(0,0)[b]{23}}
\put(128,134){\makebox(0,0)[b]{24}}
\put(166,0){\makebox(0,0)[l]{25}}
\put(166,32){\makebox(0,0)[l]{26}}
\put(166,64){\makebox(0,0)[l]{27}}
\put(166,96){\makebox(0,0)[l]{28}}
\put(160,134){\makebox(0,0)[b]{29}}
\put(198,0){\makebox(0,0)[l]{30}}
\put(198,32){\makebox(0,0)[l]{31}}
\put(198,64){\makebox(0,0)[l]{32}}
\put(198,96){\makebox(0,0)[l]{33}}
\put(192,134){\makebox(0,0)[b]{34}}
\end{picture}
\caption{A non-decreasing unbounded function $g(x)$, together with the pairing function $\phi_g(x,y)$ that is constructed from $g$.  Points in the graph of $g$ are interpolated with dotted line segments.  To make the sequence of points $\phi_g^{-1}(0)$, $\phi_g^{-1}(1)$, $\phi_g^{-1}(2)$, \ldots\ more visually apparent, solid line segments have been drawn between some of the points that occur consecutively in the sequence.}
\label{phi-illust}
\end{figure}
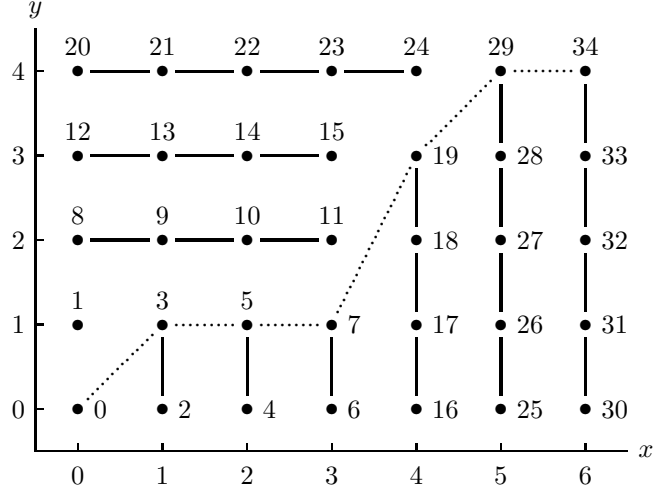

\begin{definition}\label{psi-def}
Let $g\colon\mathbb{N}\to\mathbb{N}$ be any unbounded function, and let $s_0$, $s_1$, $s_2$, \ldots\ be the sequence of step points of $g$.  Define $\psi_g\colon\mathbb{N}\to\mathbb{N}^2$ so that
\begin{equation*}
\psi_g(z)=\begin{cases}
\biggl(z\bmod s_m\;,\;\Bigl\lfloor\,\dfrac{z}{s_m}\Bigr\rfloor\biggr) &\text{\quad if $z<s_m\bigl(g(s_m)+1\bigr)$}\\
\biggl(\Bigl\lfloor\dfrac{z}{\,g(s_m)+1}\Bigr\rfloor\;,\;z\bmod\bigl(g(s_m)+1\bigr)\biggr) &\text{\quad otherwise}\rule{0pt}{21pt}
\end{cases},
\end{equation*}
where $m$ is the smallest non-negative integer such that $z\in B_m$.
\end{definition}
%
%

\begin{lemma}\label{phi-a-b}
Let $g\colon\mathbb{N}\to\mathbb{N}$ be any non-decreasing unbounded function.  For each $k\in\mathbb{N}$, $\phi_g$ is a bijection from $A_k$ to $B_k$, and $\psi_g$ is its inverse.
\end{lemma}
\begin{proof}
Consider any $k\in\mathbb{N}$.  Since $A_k$ and $B_k$ are finite sets with the same number of elements, it is sufficient to prove that $\phi_g$ is an injection from $A_k$ to $B_k$, and that $\psi_g$ is its inverse.  Equivalently, it is sufficient to prove that $\phi_g$ is a function from $A_k$ to $B_k$, and that $\psi_g$ is its left inverse.  The proof is by induction on $k$.

For the base step, consider any $(x,y)\in A_0$.  By Definition~\ref{a-b-def},
\begin{equation}\label{y-leq-g0}
y\leq g(s_0)
\end{equation}
and
\begin{equation}\label{x-less-s1}
x<s_1.
\end{equation}
Then, by Lemma~\ref{step-points}, $g(x)\leq g(s_0)$.  But $s_0=0$ by equation~\eqref{s0}, and $g$ is non-decreasing, so $g(x)<g(s_0)$ implies that $x<0$.  Since $x$ cannot be negative, it must be the case that $g(x)=g(s_0)$.  Therefore, by inequality~\eqref{y-leq-g0}, $y\leq g(x)$.  It immediately follows from Definition~\ref{phi-def} that
\begin{equation}\label{quot-rem-gs0}
\phi_g(x,y)=x\bigl(g(s_0)+1\bigr)+y.
\end{equation}
Together with inequalities~\eqref{y-leq-g0} and~\eqref{x-less-s1}, this implies that
\begin{align*}
\phi_g(x,y)&\leq(s_1-1)\bigl(g(s_0)+1\bigr)+g(s_0),\\
\phi_g(x,y)&\leq s_1\bigl(g(s_0)+1\bigr)-g(s_0)-1+g(s_0),\\
\phi_g(x,y)&\leq s_1\bigl(g(s_0)+1\bigr)-1.
\end{align*}
Therefore, $\phi_g(x,y)\in B_0$ by Definition~\ref{a-b-def}.  We have shown that $\phi_g$ is a function from $A_0$ to $B_0$.

Now, we have also shown that $m=0$ is the smallest non-negative integer such that $\phi_g(x,y)\in B_m$.  Furthermore, $\phi_g(x,y)\geq s_0\bigl(g(s_0)+1\bigr)$ because $s_0=0$.  Therefore, by Definition~\ref{psi-def},
\begin{equation*}
\psi_g\bigl(\phi_g(x,y)\bigr)=\biggl(\Bigl\lfloor\dfrac{\phi_g(x,y)}{\,g(s_0)+1}\Bigr\rfloor\;,\;\phi_g(x,y)\bmod\bigl(g(s_0)+1\bigr)\biggr).
\end{equation*}
But $y<g(s_0)+1$ by inequality~\eqref{y-leq-g0}.  So, by equation~\eqref{quot-rem-gs0} and the Euclidean division theorem, $\psi_g\bigl(\phi_g(x,y)\bigr)=(x,y)$.  We have shown that $\psi_g$ is a left inverse of $\phi_g\colon A_0\to B_0$.

Next, consider any $k\in\mathbb{N}$ and suppose, as the induction hypothesis, that $\phi_g$ is a function from $A_k$ to $B_k$, and that $\psi_g\bigl(\phi_g(x,y)\bigr)=(x,y)$ for all $(x,y)\in A_k$.  Now consider any $(x,y)\in A_{k+1}$.  By Definition~\ref{a-b-def},
\begin{equation}\label{x-less-sk2}
x<s_{k+2}
\end{equation}
and
\begin{equation}\label{y-leq-gk1}
y\leq g(s_{k+1}).
\end{equation}
If $(x,y)\in A_k$, then it follows from the induction hypothesis that $\phi_g(x,y)\in B_k\subseteq B_{k+1}$ and $\psi_g\bigl(\phi_g(x,y)\bigr)=(x,y)$.  Otherwise, $(x,y)\notin A_k$ and Definition~\ref{a-b-def} implies that
\begin{equation}\label{x-y-notin-ak}
x\geq s_{k+1}\quad\text{or}\quad y>g(s_k).
\end{equation}
There are two cases to consider.
\begin{description}
\item[Case 1] If
\begin{equation}\label{x-geq-sk1}
x\geq s_{k+1}
\end{equation}
then $g(x)\geq g(s_{k+1})$ because $g$ is non-decreasing.  But by inequality~\eqref{x-less-sk2} and Lemma~\ref{step-points}, $g(x)\leq g(s_{k+1})$.  Therefore, $g(x)=g(s_{k+1})$.  And by inequality~\eqref{y-leq-gk1}, this implies that $y\leq g(x)$.  It then follows from Definition~\ref{phi-def} that
\begin{equation}\label{quot-rem-gsk1}
\phi_g(x,y)=x\bigl(g(s_{k+1})+1\bigr)+y.
\end{equation}
Together with inequalities~\eqref{x-less-sk2} and~\eqref{y-leq-gk1}, this implies that
\begin{align*}
\phi_g(x,y)&\leq(s_{k+2}-1)\bigl(g(s_{k+1})+1\bigr)+g(s_{k+1}),\\
\phi_g(x,y)&\leq s_{k+2}\bigl(g(s_{k+1})+1\bigr)-g(s_{k+1})-1+g(s_{k+1}),\\
\phi_g(x,y)&\leq s_{k+2}\bigl(g(s_{k+1})+1\bigr)-1.
\end{align*}
Therefore, $\phi_g(x,y)\in B_{k+1}$ by Definition~\ref{a-b-def}.  Moreover, by inequality~\eqref{x-geq-sk1} and equation~\eqref{quot-rem-gsk1},
\begin{equation}\label{psi-cond2}
\phi_g(x,y)\geq s_{k+1}\bigl(g(s_{k+1})+1\bigr).
\end{equation}
But $s_{k+1}\geq s_k$, so $g(s_{k+1})\geq g(s_k)$.  This implies that
\begin{equation*}
\phi_g(x,y)\geq s_{k+1}\bigl(g(s_k)+1\bigr).
\end{equation*}
Hence, $\phi_g(x,y)\notin B_k$.  It immediately follows that $m=k+1$ is the smallest non-negative integer such that $\phi_g(x,y)\in B_m$.  Therefore, by inequality~\eqref{psi-cond2} and Definition~\ref{psi-def},
\begin{equation*}
\psi_g\bigl(\phi_g(x,y)\bigr)=\biggl(\Bigl\lfloor\dfrac{\phi_g(x,y)}{\,g(s_{k+1})+1}\Bigr\rfloor\;,\;\phi_g(x,y)\bmod\bigl(g(s_{k+1})+1\bigr)\biggr).
\end{equation*}
But $y<g(s_{k+1})+1$ by inequality~\eqref{y-leq-gk1}.  So, by equation~\eqref{quot-rem-gsk1} and the Euclidean division theorem, $\psi_g\bigl(\phi_g(x,y)\bigr)=(x,y)$.
\item[Case 2] If
\begin{equation}\label{x-less-sk1}
x<s_{k+1}
\end{equation}
then $g(x)\leq g(s_k)$ by Lemma~\ref{step-points}.  But by condition~\eqref{x-y-notin-ak},
\begin{equation}\label{y-greater-gsk}
y>g(s_k).
\end{equation}
So, $y>g(s_k)\geq g(x)$.  It immediately follows from Definition~\ref{phi-def} that
\begin{equation}\label{quot-rem-pseudo-inv-y}
\phi_g(x,y)=y\cdot g^+(y)+x.
\end{equation}
Now consider any $u\in\mathbb{N}$ and suppose that $u<s_{k+1}$.  By Lemma~\ref{step-points} and inequality~\eqref{y-greater-gsk}, $g(u)\leq g(s_k)<y$.  So,
\begin{equation*}
u<s_{k+1}\quad\text{implies}\quad g(u)<y.
\end{equation*}
The contrapositive of this statement is
\begin{equation*}
g(u)\geq y\quad\text{implies}\quad u\geq s_{k+1}.
\end{equation*}
Together with inequality~\eqref{y-leq-gk1}, this implies that
\begin{equation*}
g(s_{k+1})\geq y\quad\And\quad\forall u\in\mathbb{N}\,\bigl(\,g(u)\geq y\quad\text{implies}\quad u\geq s_{k+1}\,\bigr).
\end{equation*}
By condition~\eqref{pseudo-inv}, $g^+(y)=s_{k+1}$.  Hence, by equation~\eqref{quot-rem-pseudo-inv-y},
\begin{equation}\label{quot-rem-sk1}
\phi_g(x,y)=y\cdot s_{k+1}+x.
\end{equation}
Then, by inequalities~\eqref{y-leq-gk1} and~\eqref{x-less-sk1},
\begin{align}
\phi_g(x,y)&<g(s_{k+1})\cdot s_{k+1}+s_{k+1},\notag\\
\phi_g(x,y)&<s_{k+1}\bigl(g(s_{k+1})+1\bigr).\label{psi-cond1}
\end{align}
But $s_{k+1}<s_{k+2}$, so $\phi_g(x,y)<s_{k+2}\bigl(g(s_{k+1})+1\bigr)$.  It immediately follows from Definition~\ref{a-b-def} that $\phi_g(x,y)\in B_{k+1}$.  But by inequality~\eqref{y-greater-gsk}, $y\geq g(s_k)+1$.  Therefore, by equation~\eqref{quot-rem-sk1}, $\phi_g(x,y)\geq s_{k+1}\bigl(g(s_k)+1\bigr)$.  Then, by Definition~\ref{a-b-def}, $\phi_g(x,y)\notin B_k$.  We have shown that $m=k+1$ is the smallest non-negative integer such that $\phi_g(x,y)\in B_m$.  Therefore, by inequality~\eqref{psi-cond1} and Definition~\ref{psi-def},
\begin{equation*}
\psi_g\bigl(\phi_g(x,y)\bigr)=\biggl(\phi_g(x,y)\bmod s_{k+1}\;,\;\Bigl\lfloor\,\dfrac{\phi_g(x,y)}{s_{k+1}}\Bigr\rfloor\biggr).
\end{equation*}
It then follows from inequality~\eqref{x-less-sk1}, equation~\eqref{quot-rem-sk1}, and the Euclidean division theorem that $\psi_g\bigl(\phi_g(x,y)\bigr)=(x,y)$.
\end{description}
In all cases, we have shown that $\phi_g(x,y)\in B_{k+1}$ and $\psi_g\bigl(\phi_g(x,y)\bigr)=(x,y)$.  Therefore, $\phi_g$ is a function from $A_{k+1}$ to $B_{k+1}$, and $\psi_g$ is its left inverse.
\end{proof}

\begin{theorem}\label{phi-pairing}
Let $g\colon\mathbb{N}\to\mathbb{N}$ be any non-decreasing unbounded function.  Then $\phi_g$ is a pairing function from $\mathbb{N}^2$ to $\mathbb{N}$, and $\psi_g$ is its inverse.
\end{theorem}
\begin{proof}
Consider any $(x,y)\in\mathbb{N}^2$, and note that there must exist a non-negative integer $k$ such that $(x,y)\in A_k$.  Therefore, by Lemma~\ref{phi-a-b}, $\psi_g\bigl(\phi_g(x,y)\bigr)=(x,y)$.  Similarly, consider any $z\in\mathbb{N}$.  There must exist a non-negative integer $k^\prime$ such that $z\in B_{k^\prime}$.  Therefore, by Lemma~\ref{phi-a-b}, $\phi_g\bigl(\psi_g(z)\bigr)=z$.  It immediately follows that $\psi_g\colon\mathbb{N}\to\mathbb{N}^2$ is the unique inverse of the bijection $\phi_g\colon\mathbb{N}^2\to\mathbb{N}$.
\end{proof}

\begin{corollary}\label{a-b-elem}
Let $g\colon\mathbb{N}\to\mathbb{N}$ be any non-decreasing unbounded function.  For all non-negative integers $x$, $y$, and $k$,
\begin{equation*}
(x,y)\in A_k\quad\text{if and only if}\quad\phi_g(x,y)\in B_k.
\end{equation*}
\end{corollary}
\begin{proof}
Consider any non-negative integers $x$, $y$, and $k$.  Now suppose that $(x,y)\in A_k$.  By Lemma~\ref{phi-a-b}, $\phi_g(x,y)\in B_k$.  Conversely, suppose that $\phi_g(x,y)\in B_k$.  By Lemma~\ref{phi-a-b}, $\psi_g\bigl(\phi_g(x,y)\bigr)\in A_k$.  But by Theorem~\ref{phi-pairing}, $\psi_g\bigl(\phi_g(x,y)\bigr)=(x,y)$.  Hence, $(x,y)\in A_k$.
\end{proof}

Given any non-decreasing unbounded function $g\colon\mathbb{N}\to\mathbb{N}$, and given any $x,y\in\mathbb{N}$, define $\mu(x,y)$ to be the smallest non-negative integer $k$ such that $\phi_g(x,y)\in B_k$.  By this definition, the quantity $m$ that appears in Definition~\ref{psi-def} is equal to $\mu\bigl(\psi_g(z)\bigr)$.  And by Corollary~\ref{a-b-elem}, $\mu(x,y)$ is also the smallest non-negative integer $k$ such that $(x,y)\in A_k$.  It then follows from equations~\eqref{u-def} and~\eqref{j-def} that $A_k=U_\mu^{<k+1}$ and $B_k=J_\mu^{<k+1}$ for each $k\in\mathbb{N}$.  

\begin{theorem}\label{phi-shell-numbering}
Let $g\colon\mathbb{N}\to\mathbb{N}$ be any non-decreasing unbounded function.  Then $\mu$ is a shell numbering for $\phi_g$.
%
%
\end{theorem}
\begin{proof}
Consider any points $(x_1,y_1)$ and $(x_2,y_2)$ in $\mathbb{N}^2$.  Let $m_1=\mu(x_1,y_1)$ and let $m_2=\mu(x_2,y_2)$.  Now suppose that $\mu(x_1,y_1)<\mu(x_2,y_2)$.  Then $m_1<m_2$.  Since $m_2$ is the smallest non-negative integer such that $\phi_g(x_2,y_2)\in B_{m_2}$, and since $m_1$ is smaller than $m_2$, it must be the case that $\phi_g(x_2,y_2)\notin B_{m_1}$.  But $\phi_g(x_1,y_1)\in B_{m_1}$ by the definition of $\mu(x_1,y_1)$.  Hence, by Definition~\ref{a-b-def},
\begin{equation*}
\phi_g(x_1,y_1)<s_{m_1+1}\bigl(g(s_{m_1})+1\bigr)\leq\phi_g(x_2,y_2).
\end{equation*}
We have shown that $\mu(x_1,y_1)<\mu(x_2,y_2)$ implies $\phi_g(x_1,y_1)<\phi_g(x_2,y_2)$.  Therefore, by condition~\eqref{shell-def}, $\mu$ is shell numbering for $\phi_g$.
\end{proof}

\section{A Proportional Pairing Function}\label{derivation}

Given any positive integers $a$ and $b$, and any non-negative integer $x$, define\footnote{
Many of the formulas in this section that are expressed with the floor operation can also be expressed with the ceiling operation if one uses the fact that $\bigl\lfloor\sqrt[a]{x}\,\bigr\rfloor+1=\bigl\lceil\sqrt[a]{x+1}\;\bigr\rceil$ for all positive integers $a$ and all non-negative integers $x$.
%
%
}
\begin{equation*}
g_{a,b}(x)=\Bigl(\bigl\lfloor\sqrt[a]{x}\,\bigr\rfloor+1\Bigr)^b-1.
\end{equation*}
Note that $g_{a,b}$ is a non-decreasing unbounded function from $\mathbb{N}$ to $\mathbb{N}$.  In this section, we will use the function $g_{a,b}$ to construct the pairing function $p_{a,b}$.  And in the course of deriving the formula for $p_{a,b}$, we will rely on the following well-known identity~\cite{Graham1994}, which holds for all integers $k$ and all real numbers $t$:
\begin{equation}\label{k-leq-t}
k\leq t\quad\text{if and only if}\quad k\leq\lfloor t\rfloor.
\end{equation}

\begin{lemma}\label{y-greater-g}
Let $a$ and $b$ be any positive integers.  For all non-negative integers $x$ and $y$,
\begin{equation*}
y>g_{a,b}(x)\quad\text{if and only if}\quad\bigl\lfloor\sqrt[b]{y\rule{0pt}{6pt}}\,\bigr\rfloor>\bigl\lfloor\sqrt[a]{x}\,\bigr\rfloor.
\end{equation*}
\end{lemma}
\begin{proof}
Consider any non-negative integers $x$ and $y$.  The following statements are equivalent:
\begin{alignat*}{2}
%
%
y&>{}&g_{a,b}(x),&\\
y&>{}&\Bigl(\bigl\lfloor\sqrt[a]{x}\,\bigr\rfloor&+1\Bigr)^b-1,\\
y&\geq{}&\Bigl(\bigl\lfloor\sqrt[a]{x}\,\bigr\rfloor&+1\Bigr)^b,\\
\sqrt[b]{y\rule{0pt}{6pt}}&\geq{}&\bigl\lfloor\sqrt[a]{x}\,\bigr\rfloor&+1.
\end{alignat*}
But by condition~\eqref{k-leq-t}, this is equivalent to
\begin{align*}
\bigl\lfloor\sqrt[b]{y\rule{0pt}{6pt}}\,\bigr\rfloor&\geq\bigl\lfloor\sqrt[a]{x}\,\bigr\rfloor+1,\\
\bigl\lfloor\sqrt[b]{y\rule{0pt}{6pt}}\,\bigr\rfloor&>\bigl\lfloor\sqrt[a]{x}\,\bigr\rfloor.
\end{align*}
\end{proof}

\begin{lemma}\label{pseudo-inv-floor}
Let $a$ and $b$ be any positive integers.  Then $g_{a,b}^+(y)=\bigl\lfloor\sqrt[b]{y\rule{0pt}{6pt}}\,\bigr\rfloor^a$ for all non-negative integers $y$.
\end{lemma}
\begin{proof}
Consider any non-negative integer $y$.  By definition, $g_{a,b}^+(y)$ is the smallest non-negative integer $x$ such that $g_{a,b}(x)\geq y$.  But by Lemma~\ref{y-greater-g}, this inequality is equivalent to $\bigl\lfloor\sqrt[b]{y\rule{0pt}{6pt}}\,\bigr\rfloor\leq\bigl\lfloor\sqrt[a]{x}\,\bigr\rfloor$.  And by condition~\eqref{k-leq-t}, this is equivalent to
\begin{alignat*}{2}
&\bigl\lfloor\sqrt[b]{y\rule{0pt}{6pt}}\,\bigr\rfloor&&\leq\sqrt[a]{x},\\
&\bigl\lfloor\sqrt[b]{y\rule{0pt}{6pt}}\,\bigr\rfloor^a&&\leq x.
\end{alignat*}
Since $g_{a,b}^+(y)$ is the smallest non-negative integer $x$ satisfying this inequality, it must be the case that $g_{a,b}^+(y)=\bigl\lfloor\sqrt[b]{y\rule{0pt}{6pt}}\,\bigr\rfloor^a$.
\end{proof}

Now, for each pair $a$ and $b$ of positive integers, let $p_{a,b}$ be the function $\phi_{g_{a,b}}$ that is defined in Definition~\ref{phi-def}.  Then by Lemmas~\ref{y-greater-g} and~\ref{pseudo-inv-floor},
\begin{equation*}
p_{a,b}(x,y)=\begin{cases}
y\bigl\lfloor\sqrt[b]{y\rule{0pt}{6pt}}\,\bigr\rfloor^a+x &\text{\quad if $\bigl\lfloor\sqrt[b]{y\rule{0pt}{6pt}}\,\bigr\rfloor>\bigl\lfloor\sqrt[a]{x}\,\bigr\rfloor$}\\
x\Bigl(\bigl\lfloor\sqrt[a]{x}\,\bigr\rfloor+1\Bigr)^{\!b}+y &\text{\quad otherwise}\rule{0pt}{20pt}
\end{cases}
\end{equation*}
for all $(x,y)\in\mathbb{N}^2$.  It then follows from Theorem~\ref{phi-pairing} that $p_{a,b}$ is a pairing function.  And when $a=1$ and $b=1$, we have
\begin{equation*}
p_{1,1}(x,y)=\begin{cases}
y^2+x &\text{\quad if $y>x$}\\
x^2+x+y &\text{\quad otherwise}\rule{0pt}{15pt}
\end{cases}.
\end{equation*}
This special case has previously been considered by Rosenberg~\cite{Rosenberg1978}.  It has also been studied by Szudzik~\cite{Szudzik2006}.  The function $p_{3,2}$ is illustrated in Figure~\ref{p-3-2},
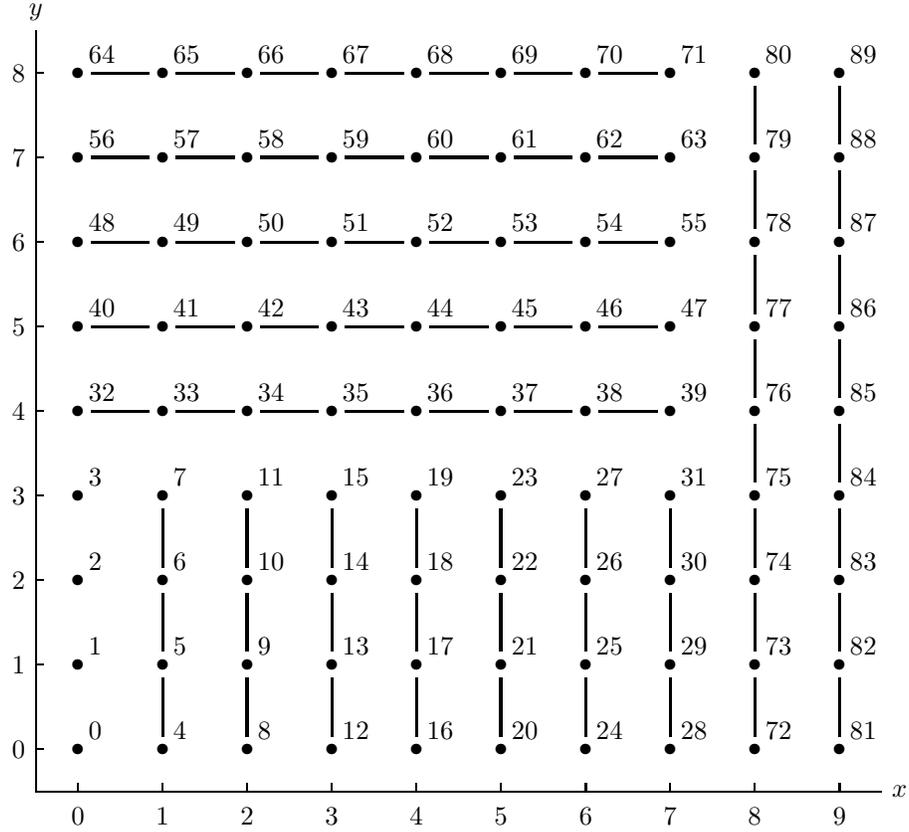
\begin{figure}
\centering
\begin{picture}(339,312)(-25,-29)
%
%
\multiput(0,0)(0,32){9}{\multiput(0,0)(32,0){10}{\circle*{4}}}
%
\thinlines
\put(-16,-16){\line(1,0){320}}
\put(-16,-16){\line(0,1){288}}
%
\put(308,-16){\makebox(0,0)[l]{$x$}}
\put(-16,276){\makebox(0,0)[b]{$y$}}
%
\thinlines
\multiput(0,-16)(32,0){10}{\line(0,1){3}}
\multiput(-16,0)(0,32){9}{\line(1,0){3}}
%
\put(0,-22){\makebox(0,0)[t]{0}}
\put(32,-22){\makebox(0,0)[t]{1}}
\put(64,-22){\makebox(0,0)[t]{2}}
\put(96,-22){\makebox(0,0)[t]{3}}
\put(128,-22){\makebox(0,0)[t]{4}}
\put(160,-22){\makebox(0,0)[t]{5}}
\put(192,-22){\makebox(0,0)[t]{6}}
\put(224,-22){\makebox(0,0)[t]{7}}
\put(256,-22){\makebox(0,0)[t]{8}}
\put(288,-22){\makebox(0,0)[t]{9}}
%
\put(-20,0){\makebox(0,0)[r]{0}}
\put(-20,32){\makebox(0,0)[r]{1}}
\put(-20,64){\makebox(0,0)[r]{2}}
\put(-20,96){\makebox(0,0)[r]{3}}
\put(-20,128){\makebox(0,0)[r]{4}}
\put(-20,160){\makebox(0,0)[r]{5}}
\put(-20,192){\makebox(0,0)[r]{6}}
\put(-20,224){\makebox(0,0)[r]{7}}
\put(-20,256){\makebox(0,0)[r]{8}}
%
\thicklines
\put(32,5){\line(0,1){22}}
\put(32,37){\line(0,1){22}}
\put(32,69){\line(0,1){22}}
\put(64,5){\line(0,1){22}}
\put(64,37){\line(0,1){22}}
\put(64,69){\line(0,1){22}}
\put(96,5){\line(0,1){22}}
\put(96,37){\line(0,1){22}}
\put(96,69){\line(0,1){22}}
\put(128,5){\line(0,1){22}}
\put(128,37){\line(0,1){22}}
\put(128,69){\line(0,1){22}}
\put(160,5){\line(0,1){22}}
\put(160,37){\line(0,1){22}}
\put(160,69){\line(0,1){22}}
\put(192,5){\line(0,1){22}}
\put(192,37){\line(0,1){22}}
\put(192,69){\line(0,1){22}}
\put(224,5){\line(0,1){22}}
\put(224,37){\line(0,1){22}}
\put(224,69){\line(0,1){22}}
\put(5,128){\line(1,0){22}}
\put(37,128){\line(1,0){22}}
\put(69,128){\line(1,0){22}}
\put(101,128){\line(1,0){22}}
\put(133,128){\line(1,0){22}}
\put(165,128){\line(1,0){22}}
\put(197,128){\line(1,0){22}}
\put(5,160){\line(1,0){22}}
\put(37,160){\line(1,0){22}}
\put(69,160){\line(1,0){22}}
\put(101,160){\line(1,0){22}}
\put(133,160){\line(1,0){22}}
\put(165,160){\line(1,0){22}}
\put(197,160){\line(1,0){22}}
\put(5,192){\line(1,0){22}}
\put(37,192){\line(1,0){22}}
\put(69,192){\line(1,0){22}}
\put(101,192){\line(1,0){22}}
\put(133,192){\line(1,0){22}}
\put(165,192){\line(1,0){22}}
\put(197,192){\line(1,0){22}}
\put(5,224){\line(1,0){22}}
\put(37,224){\line(1,0){22}}
\put(69,224){\line(1,0){22}}
\put(101,224){\line(1,0){22}}
\put(133,224){\line(1,0){22}}
\put(165,224){\line(1,0){22}}
\put(197,224){\line(1,0){22}}
\put(5,256){\line(1,0){22}}
\put(37,256){\line(1,0){22}}
\put(69,256){\line(1,0){22}}
\put(101,256){\line(1,0){22}}
\put(133,256){\line(1,0){22}}
\put(165,256){\line(1,0){22}}
\put(197,256){\line(1,0){22}}
\put(256,5){\line(0,1){22}}
\put(256,37){\line(0,1){22}}
\put(256,69){\line(0,1){22}}
\put(256,101){\line(0,1){22}}
\put(256,133){\line(0,1){22}}
\put(256,165){\line(0,1){22}}
\put(256,197){\line(0,1){22}}
\put(256,229){\line(0,1){22}}
\put(288,5){\line(0,1){22}}
\put(288,37){\line(0,1){22}}
\put(288,69){\line(0,1){22}}
\put(288,101){\line(0,1){22}}
\put(288,133){\line(0,1){22}}
\put(288,165){\line(0,1){22}}
\put(288,197){\line(0,1){22}}
\put(288,229){\line(0,1){22}}
%
\put(4,4){\makebox(0,0)[bl]{0}}
\put(4,36){\makebox(0,0)[bl]{1}}
\put(4,68){\makebox(0,0)[bl]{2}}
\put(4,100){\makebox(0,0)[bl]{3}}
\put(36,4){\makebox(0,0)[bl]{4}}
\put(36,36){\makebox(0,0)[bl]{5}}
\put(36,68){\makebox(0,0)[bl]{6}}
\put(36,100){\makebox(0,0)[bl]{7}}
\put(68,4){\makebox(0,0)[bl]{8}}
\put(68,36){\makebox(0,0)[bl]{9}}
\put(68,68){\makebox(0,0)[bl]{10}}
\put(68,100){\makebox(0,0)[bl]{11}}
\put(100,4){\makebox(0,0)[bl]{12}}
\put(100,36){\makebox(0,0)[bl]{13}}
\put(100,68){\makebox(0,0)[bl]{14}}
\put(100,100){\makebox(0,0)[bl]{15}}
\put(132,4){\makebox(0,0)[bl]{16}}
\put(132,36){\makebox(0,0)[bl]{17}}
\put(132,68){\makebox(0,0)[bl]{18}}
\put(132,100){\makebox(0,0)[bl]{19}}
\put(164,4){\makebox(0,0)[bl]{20}}
\put(164,36){\makebox(0,0)[bl]{21}}
\put(164,68){\makebox(0,0)[bl]{22}}
\put(164,100){\makebox(0,0)[bl]{23}}
\put(196,4){\makebox(0,0)[bl]{24}}
\put(196,36){\makebox(0,0)[bl]{25}}
\put(196,68){\makebox(0,0)[bl]{26}}
\put(196,100){\makebox(0,0)[bl]{27}}
\put(228,4){\makebox(0,0)[bl]{28}}
\put(228,36){\makebox(0,0)[bl]{29}}
\put(228,68){\makebox(0,0)[bl]{30}}
\put(228,100){\makebox(0,0)[bl]{31}}
\put(4,132){\makebox(0,0)[bl]{32}}
\put(36,132){\makebox(0,0)[bl]{33}}
\put(68,132){\makebox(0,0)[bl]{34}}
\put(100,132){\makebox(0,0)[bl]{35}}
\put(132,132){\makebox(0,0)[bl]{36}}
\put(164,132){\makebox(0,0)[bl]{37}}
\put(196,132){\makebox(0,0)[bl]{38}}
\put(228,132){\makebox(0,0)[bl]{39}}
\put(4,164){\makebox(0,0)[bl]{40}}
\put(36,164){\makebox(0,0)[bl]{41}}
\put(68,164){\makebox(0,0)[bl]{42}}
\put(100,164){\makebox(0,0)[bl]{43}}
\put(132,164){\makebox(0,0)[bl]{44}}
\put(164,164){\makebox(0,0)[bl]{45}}
\put(196,164){\makebox(0,0)[bl]{46}}
\put(228,164){\makebox(0,0)[bl]{47}}
\put(4,196){\makebox(0,0)[bl]{48}}
\put(36,196){\makebox(0,0)[bl]{49}}
\put(68,196){\makebox(0,0)[bl]{50}}
\put(100,196){\makebox(0,0)[bl]{51}}
\put(132,196){\makebox(0,0)[bl]{52}}
\put(164,196){\makebox(0,0)[bl]{53}}
\put(196,196){\makebox(0,0)[bl]{54}}
\put(228,196){\makebox(0,0)[bl]{55}}
\put(4,228){\makebox(0,0)[bl]{56}}
\put(36,228){\makebox(0,0)[bl]{57}}
\put(68,228){\makebox(0,0)[bl]{58}}
\put(100,228){\makebox(0,0)[bl]{59}}
\put(132,228){\makebox(0,0)[bl]{60}}
\put(164,228){\makebox(0,0)[bl]{61}}
\put(196,228){\makebox(0,0)[bl]{62}}
\put(228,228){\makebox(0,0)[bl]{63}}
\put(4,260){\makebox(0,0)[bl]{64}}
\put(36,260){\makebox(0,0)[bl]{65}}
\put(68,260){\makebox(0,0)[bl]{66}}
\put(100,260){\makebox(0,0)[bl]{67}}
\put(132,260){\makebox(0,0)[bl]{68}}
\put(164,260){\makebox(0,0)[bl]{69}}
\put(196,260){\makebox(0,0)[bl]{70}}
\put(228,260){\makebox(0,0)[bl]{71}}
\put(260,4){\makebox(0,0)[bl]{72}}
\put(260,36){\makebox(0,0)[bl]{73}}
\put(260,68){\makebox(0,0)[bl]{74}}
\put(260,100){\makebox(0,0)[bl]{75}}
\put(260,132){\makebox(0,0)[bl]{76}}
\put(260,164){\makebox(0,0)[bl]{77}}
\put(260,196){\makebox(0,0)[bl]{78}}
\put(260,228){\makebox(0,0)[bl]{79}}
\put(260,260){\makebox(0,0)[bl]{80}}
\put(292,4){\makebox(0,0)[bl]{81}}
\put(292,36){\makebox(0,0)[bl]{82}}
\put(292,68){\makebox(0,0)[bl]{83}}
\put(292,100){\makebox(0,0)[bl]{84}}
\put(292,132){\makebox(0,0)[bl]{85}}
\put(292,164){\makebox(0,0)[bl]{86}}
\put(292,196){\makebox(0,0)[bl]{87}}
\put(292,228){\makebox(0,0)[bl]{88}}
\put(292,260){\makebox(0,0)[bl]{89}}
\end{picture}
\caption{The pairing function $p_{3,2}(x,y)$.  To make the sequence of points $p_{3,2}^{-1}(0)$, $p_{3,2}^{-1}(1)$, $p_{3,2}^{-1}(2)$, \ldots\ more visually apparent, line segments have been drawn between some of the points that occur consecutively in the sequence.}
\label{p-3-2}
\end{figure}
and $p_{1,1}$ is illustrated in Figure~\ref{p-1-1}.
\begin{figure}
\centering
\begin{picture}(179,184)(-25,-29)
%
%
\multiput(0,0)(0,32){5}{\multiput(0,0)(32,0){5}{\circle*{4}}}
%
\thinlines
\put(-16,-16){\line(1,0){160}}
\put(-16,-16){\line(0,1){160}}
%
\put(148,-16){\makebox(0,0)[l]{$x$}}
\put(-16,148){\makebox(0,0)[b]{$y$}}
%
\thinlines
\multiput(0,-16)(32,0){5}{\line(0,1){3}}
\multiput(-16,0)(0,32){5}{\line(1,0){3}}
%
\put(0,-22){\makebox(0,0)[t]{0}}
\put(32,-22){\makebox(0,0)[t]{1}}
\put(64,-22){\makebox(0,0)[t]{2}}
\put(96,-22){\makebox(0,0)[t]{3}}
\put(128,-22){\makebox(0,0)[t]{4}}
%
\put(-20,0){\makebox(0,0)[r]{0}}
\put(-20,32){\makebox(0,0)[r]{1}}
\put(-20,64){\makebox(0,0)[r]{2}}
\put(-20,96){\makebox(0,0)[r]{3}}
\put(-20,128){\makebox(0,0)[r]{4}}
%
\thicklines
\put(32,5){\line(0,1){22}}
\put(5,64){\line(1,0){22}}
\put(64,5){\line(0,1){22}}
\put(64,37){\line(0,1){22}}
\put(5,96){\line(1,0){22}}
\put(37,96){\line(1,0){22}}
\put(96,5){\line(0,1){22}}
\put(96,37){\line(0,1){22}}
\put(96,69){\line(0,1){22}}
\put(5,128){\line(1,0){22}}
\put(37,128){\line(1,0){22}}
\put(69,128){\line(1,0){22}}
\put(128,5){\line(0,1){22}}
\put(128,37){\line(0,1){22}}
\put(128,69){\line(0,1){22}}
\put(128,101){\line(0,1){22}}
%
\put(4,4){\makebox(0,0)[bl]{0}}
\put(4,36){\makebox(0,0)[bl]{1}}
\put(36,4){\makebox(0,0)[bl]{2}}
\put(36,36){\makebox(0,0)[bl]{3}}
\put(4,68){\makebox(0,0)[bl]{4}}
\put(36,68){\makebox(0,0)[bl]{5}}
\put(68,4){\makebox(0,0)[bl]{6}}
\put(68,36){\makebox(0,0)[bl]{7}}
\put(68,68){\makebox(0,0)[bl]{8}}
\put(4,100){\makebox(0,0)[bl]{9}}
\put(36,100){\makebox(0,0)[bl]{10}}
\put(68,100){\makebox(0,0)[bl]{11}}
\put(100,4){\makebox(0,0)[bl]{12}}
\put(100,36){\makebox(0,0)[bl]{13}}
\put(100,68){\makebox(0,0)[bl]{14}}
\put(100,100){\makebox(0,0)[bl]{15}}
\put(4,132){\makebox(0,0)[bl]{16}}
\put(36,132){\makebox(0,0)[bl]{17}}
\put(68,132){\makebox(0,0)[bl]{18}}
\put(100,132){\makebox(0,0)[bl]{19}}
\put(132,4){\makebox(0,0)[bl]{20}}
\put(132,36){\makebox(0,0)[bl]{21}}
\put(132,68){\makebox(0,0)[bl]{22}}
\put(132,100){\makebox(0,0)[bl]{23}}
\put(132,132){\makebox(0,0)[bl]{24}}
\end{picture}
\caption{The pairing function $p_{1,1}(x,y)$.  To make the sequence of points $p_{1,1}^{-1}(0)$, $p_{1,1}^{-1}(1)$, $p_{1,1}^{-1}(2)$, \ldots\ more visually apparent, line segments have been drawn between some of the points that occur consecutively in the sequence.}
\label{p-1-1}
\end{figure}

Now let $s_0$, $s_1$, $s_2$, \ldots\ be the sequence of step points of $g_{a,b}$.  Notice that for each non-negative integer $k$, the equation $\bigl\lfloor\sqrt[b]{y\rule{0pt}{6pt}}\,\bigr\rfloor=k$ has a solution for $y$, namely $y=k^b$.  And since $\bigl\lfloor\sqrt[b]{y\rule{0pt}{6pt}}\,\bigr\rfloor$ is a non-negative integer for each $y\in\mathbb{N}$, the range of $g_{a,b}^+(y)=\bigl\lfloor\sqrt[b]{y\rule{0pt}{6pt}}\,\bigr\rfloor^a$ is the set $S=\bigl\{k^a:k\in\mathbb{N}\bigr\}$.  Therefore, by definition,
\begin{equation*}
s_k=k^a
\end{equation*}
for each $k\in\mathbb{N}$.  It immediately follows that
\begin{equation*}
g_{a,b}(s_k)=(k+1)^b-1
\end{equation*}
and
\begin{equation*}
s_{k+1}\bigl(g_{a,b}(s_k)+1\bigr)=(k+1)^a(k+1)^b=(k+1)^{a+b}.
\end{equation*}
Then by Definition~\ref{a-b-def},
\begin{equation}\label{p-a-def}
A_k=\bigl\{\,(x,y)\in\mathbb{N}^2\;:\;x<(k+1)^a\;\;\And\;\;y<(k+1)^b\,\bigr\}
\end{equation}
and
\begin{equation}\label{p-b-def}
B_k=\bigl\{0,1,2,\ldots,(k+1)^{a+b}-1\bigr\}.
\end{equation}

\begin{lemma}
For each $z\in\mathbb{N}$, $\bigl\lfloor z^{1/(a+b)}\bigr\rfloor$ is the smallest non-negative integer $m$ such that $z\in B_m$.
\end{lemma}
\begin{proof}
Consider any $z\in\mathbb{N}$, and let $m$ be the smallest non-negative integer such that $z\in B_m$.  Note that $z\in B_m$ if and only if
\begin{alignat*}{2}
z&<{}&(&m+1)^{a+b},\\
z^{1/(a+b)}&<{}&&m+1.
\end{alignat*}
And by condition~\eqref{k-leq-t}, this is equivalent to
\begin{align*}
\bigl\lfloor z^{1/(a+b)}\bigr\rfloor&<m+1,\\
\bigl\lfloor z^{1/(a+b)}\bigr\rfloor&\leq m.
\end{align*}
Since $m$ is the smallest non-negative integer satisfying this inequality, it must be the case that $m=\bigl\lfloor z^{1/(a+b)}\bigr\rfloor$.
\end{proof}

It immediately follows from Definition~\ref{psi-def} and Theorem~\ref{phi-pairing} that
\begin{equation*}
p_{a,b}^{-1}(z)=\begin{cases}
\biggl(z\bmod m^a\;,\;\Bigl\lfloor\,\dfrac{z}{m^a}\Bigr\rfloor\biggr) &\text{\quad if $z<m^a(m+1)^b$}\\
\biggl(\Bigl\lfloor\dfrac{z}{(m+1)^b\rule{0pt}{9pt}}\Bigr\rfloor\;,\;z\bmod (m+1)^b\biggr) &\text{\quad otherwise}\rule{0pt}{21pt}
\end{cases},
\end{equation*}
where $m=\bigl\lfloor z^{1/(a+b)}\bigr\rfloor$ for each $z\in\mathbb{N}$.  Now consider the special case where $a=1$ and $z\geq m(m+1)^b$.  In this case, $m=\bigl\lfloor z^{1/(1+b)}\bigr\rfloor$.  So,
\begin{alignat*}{2}
z^{1/(1+b)}&<{}&&m+1,\\
z&<{}&(&m+1)^{1+b},\\
m(m+1)^b\leq z&<{}&(&m+1)^{1+b},
\end{alignat*}
and
\begin{equation*}
m\leq\dfrac{z}{(m+1)^b}<m+1.
\end{equation*}
Hence, $m=\Bigl\lfloor\dfrac{z}{(m+1)^b\rule{0pt}{9pt}}\Bigr\rfloor$ if $a=1$ and $z\geq m(m+1)^b$.  And by the definition of the $\bmod$ operation,
\begin{equation*}
z\bmod (m+1)^b=z-\Bigl\lfloor\dfrac{z}{(m+1)^b\rule{0pt}{9pt}}\Bigr\rfloor(m+1)^b=z-m(m+1)^b
\end{equation*}
in this case.  We may conclude that
\begin{equation}\label{p-inv-1-b}
p_{1,b}^{-1}(z)=\begin{cases}
\biggl(z\bmod m\;,\;\Bigl\lfloor\,\dfrac{z}{m}\,\Bigr\rfloor\biggr) &\text{\quad if $z<m(m+1)^b$}\\
\Bigl(m\;,\;z-m(m+1)^b\Bigr) &\text{\quad otherwise}\rule{0pt}{18pt}
\end{cases},
\end{equation}
where $m=\bigl\lfloor z^{1/(1+b)}\bigr\rfloor$ for each $z\in\mathbb{N}$.

Similarly, consider the special case where $b=1$ and $z<m^a(m+1)$.  In this case, $m=\bigl\lfloor z^{1/(a+1)}\bigr\rfloor$.  So,
\begin{align*}
m&\leq z^{1/(a+1)},\\
m^{a+1}&\leq z,
\end{align*}
and
\begin{alignat*}{2}
%
%
m^{a+1}&\leq\quad\negthickspace\!z&&<m^a(m+1),\\
m&\leq\dfrac{z}{m^a}&&<m+1.
\end{alignat*}
Hence, $m=\Bigl\lfloor\,\dfrac{z}{m^a}\Bigr\rfloor$ if $b=1$ and $z<m^a(m+1)$.  And by the definition of the $\bmod$ operation,
\begin{equation*}
z\bmod m^a=z-\Bigl\lfloor\,\dfrac{z}{m^a}\Bigr\rfloor m^a=z-m^{a+1}
\end{equation*}
in this case.  We may conclude that
\begin{equation}\label{p-inv-a-1}
p_{a,1}^{-1}(z)=\begin{cases}
\Bigl(z-m^{a+1}\;,\;m\Bigr) &\text{\quad if $z<m^a(m+1)$}\\
\biggl(\Bigl\lfloor\,\dfrac{z}{m+1}\,\Bigr\rfloor\;,\;z\bmod (m+1)\biggr) &\text{\quad otherwise}\rule{0pt}{21pt}
\end{cases},
\end{equation}
where $m=\bigl\lfloor z^{1/(a+1)}\bigr\rfloor$ for each $z\in\mathbb{N}$.  And combining equations~\eqref{p-inv-1-b} and~\eqref{p-inv-a-1}, we have that
\begin{equation*}
p_{1,1}^{-1}(z)=\begin{cases}
\Bigl(z-m^2\;,\;m\Bigr) &\text{\quad if $z<m(m+1)$}\\
\Bigl(m\;,\;z-m(m+1)\Bigr) &\text{\quad otherwise}\rule{0pt}{18pt}
\end{cases},
\end{equation*}
where $m=\bigl\lfloor\sqrt{z}\,\bigr\rfloor$ for each $z\in\mathbb{N}$.

\begin{theorem}\label{p-a-b-shell-numbering}
Let $a$ and $b$ be any positive integers.  Then $\max\Bigl(\bigl\lfloor\sqrt[a]{x}\,\bigr\rfloor,\bigl\lfloor\sqrt[b]{y\rule{0pt}{6pt}}\,\bigr\rfloor\Bigr)$ is a shell numbering for $p_{a,b}$.
\end{theorem}
\begin{proof}
By Theorem~\ref{phi-shell-numbering}, $\mu$ is a shell numbering for $p_{a,b}$.  But for each $(x,y)\in\mathbb{N}^2$, $\mu(x,y)$ is the smallest non-negative integer $m$ such that $(x,y)\in A_m$.  By equation~\eqref{p-a-def}, this is equivalent to each of the following conditions:
\begin{alignat*}{5}
x&<{}&(&m+1)^a\quad&&\And&\quad y&<{}&(&m+1)^b,\\
\sqrt[a]{x}&<{}&&m+1\quad&&\And&\quad\sqrt[b]{y\rule{0pt}{6pt}}&<{}&&m+1.
\end{alignat*}
And by condition~\eqref{k-leq-t}, this is equivalent to
\begin{alignat}{2}
\bigl\lfloor\sqrt[a]{x}\,\bigr\rfloor&<m+1\quad&&\And\quad\bigl\lfloor\sqrt[b]{y\rule{0pt}{6pt}}\,\bigr\rfloor<m+1,\notag\\
\bigl\lfloor\sqrt[a]{x}\,\bigr\rfloor&\leq m\quad&&\And\quad\bigl\lfloor\sqrt[b]{y\rule{0pt}{6pt}}\,\bigr\rfloor\leq m.\label{p-shells-cond}
\end{alignat}
Now, because $\bigl\lfloor\sqrt[a]{x}\,\bigr\rfloor$ and $\bigl\lfloor\sqrt[b]{y\rule{0pt}{6pt}}\,\bigr\rfloor$ are both non-negative integers, it cannot be the case that
\begin{equation*}
\bigl\lfloor\sqrt[a]{x}\,\bigr\rfloor<m\quad\And\quad\bigl\lfloor\sqrt[b]{y\rule{0pt}{6pt}}\,\bigr\rfloor<m,
\end{equation*}
because then
\begin{equation*}
\bigl\lfloor\sqrt[a]{x}\,\bigr\rfloor\leq m-1\quad\And\quad\bigl\lfloor\sqrt[b]{y\rule{0pt}{6pt}}\,\bigr\rfloor\leq m-1,
\end{equation*}
and this would contradict the fact that $m$ is the \emph{smallest} non-negative integer for which condition~\eqref{p-shells-cond} holds.  Therefore, it must be the case that $\bigl\lfloor\sqrt[a]{x}\,\bigr\rfloor=m$ or $\bigl\lfloor\sqrt[b]{y\rule{0pt}{6pt}}\,\bigr\rfloor=m$.  That is,
\begin{equation*} m=\max\Bigl(\bigl\lfloor\sqrt[a]{x}\,\bigr\rfloor,\bigl\lfloor\sqrt[b]{y\rule{0pt}{6pt}}\,\bigr\rfloor\Bigr).
\end{equation*}
We have shown that $\mu(x,y)=\max\Bigl(\bigl\lfloor\sqrt[a]{x}\,\bigr\rfloor,\bigl\lfloor\sqrt[b]{y\rule{0pt}{6pt}}\,\bigr\rfloor\Bigr)$ for each $(x,y)\in\mathbb{N}^2$.
\end{proof}

An immediate consequence of Theorem~\ref{p-a-b-shell-numbering} is that $\max(x,y)$ is a shell numbering for $p_{1,1}$.  That is, $p_{1,1}$ has square shells.  We conclude this section with the following theorem. 

\begin{theorem}
Let $a$ and $b$ be any positive integers.  For each integer $n>1$, $p_{a,b}$ is a base-$n$ proportional pairing function with constants of proportionality $a$ and $b$.
\end{theorem}
\begin{proof}
Consider any integer $n>1$.  Now consider any non-negative integers $x$, $y$, and $j$, and suppose that $\len_n(x)\leq aj$ and $\len_n(y)\leq bj$.  Then by condition~\eqref{len-identity}, $x<n^{aj}$ and $y<n^{bj}$.  It immediately follows from equation~\eqref{p-a-def} that $(x,y)\in A_k$, where $k=n^j-1$.  Then by Lemma~\ref{phi-a-b}, $p_{a,b}(x,y)\in B_k$.  Hence, by equation~\eqref{p-b-def}, $p_{a,b}(x,y)<n^{aj+bj}$.  And by condition~\eqref{len-identity}, this implies that $\len_n\bigl(p_{a,b}(x,y)\bigr)\leq aj+bj$.  We may then conclude, by Definition~\ref{proportional}, that $p_{a,b}$ is a base-$n$ proportional pairing function with constants of proportionality $a$ and $b$.
\end{proof}

\section{Acknowledgments}

We thank Fritz H. Obermeyer for suggesting that we should include a discussion of discrete space-filling curves in our work.

\bibliographystyle{amsplain}
\bibliography{BinaryProportionalPairing}

\end{document}